\documentclass[runningheads,a4paper]{llncs}

\usepackage{amsmath}
\usepackage{amssymb}
\usepackage{graphicx}
\usepackage[linesnumbered,ruled,vlined]{algorithm2e}
\usepackage{balance}
\usepackage{caption}
\usepackage{subfig}
\usepackage{color}

\newcommand{\tabincell}[2]{\begin{tabular}{@{}#1@{}}#2\end{tabular}}
\usepackage{url}

\begin{document}

\mainmatter  % start of an individual contribution

\pagestyle{empty}

\title{Specialty-Aware Task Assignment in Spatial Crowdsourcing}

\author{Tianshu Song, Feng Zhu, Ke Xu}

\institute{SKLSDE Lab, School of Computer Science and Engineering \\Beihang University, Beijing 100191, China}

\maketitle

\begin{abstract}
	With the rapid development of Mobile Internet, spatial crowdsourcing is gaining more and more attention from both academia and industry. 
	In spatial crowdsourcing, spatial tasks are sent to workers based on their locations.
	A wide kind of tasks in spatial crowdsourcing are specialty-aware, which are complex and need to be completed by workers with different skills collaboratively. 
	Existing studies on specialty-aware spatial crowdsourcing assume that each worker has a united charge when performing different tasks, no matter how many skills of her/him are used to complete the task, which is not fair and practical.
	In this paper, we study the problem of specialty-aware task assignment in spatial crowdsourcing, where each worker has fine-grained charge for each of their skills, and the goal is to maximize the total number of completed tasks based on tasks' budget and requirements on particular skills.
	The problem is proven to be NP-hard. Thus, we propose two efficient heuristics to solve the problem.
	Experiments on both synthetic and real datasets demonstrate the effectiveness and efficiency of our solutions.
\end{abstract}

%\abstract{
%These days, Online To Offline (O2O) platforms have been developing rapidly because of the popularization of smart phones and Mobile Internet. Spatial crowdsourcing, a burgeoning area in O2O market, is gaining more and more attention. It is+ a typical spatial crowdsourcing scenario in which an employer publishes a task and some workers will help him or her to accomplish it. However, most of previous work only considers the spatial information of workers and tasks, but ignores the individual variations among workers. In this paper, we raise a new problem called Software Development Team Formation (SDTF) problem, which aims to find a team of workers whose ability satisfies the requirement of the task. After showing the problem is NP-hard, we propose three greedy algorithms to approximately solve the problem. Besides, extensive experiments are conducted on synthetic and real datasets, which verify the effectiveness and efficiency of our algorithms.
%}
%\vspace{-2ex}

\section{Introduction}
%逻辑：
%第一段：crowdsourcing和spatial crowdsourcing很火
With the development mobile Internet and the blossom of sharing economy, all kinds of \textit{\textit{s}patial \textit{c}rowdsourcing (SC)} platforms become popular, where the \textit{online crowd workers} are employed by their phones to participate in and complete \textit{offline crowdsourcing tasks} in the physical world \cite{CHI13}. Typical SC platforms includes Gigwalk\footnote{www.gigwalk.com}, TaskRabbit\footnote{www.taskrabbit.com} and gMission\footnote{gmission.github.io} \cite{gMission}.

%第二段：不是所有的spatial crowdsourcing都基于简单任务，也有一些复杂任务存在
One fundamental issue in SC is task assignment, namely assigning crowdsourcing tasks to suitable crowd workers. Generally speaking, there are two kinds of tasks. The first kind is micro tasks which can be completed by any single worker such as taking photos and delivering things. The second kind is specialty-aware tasks such as repairing a house and organizing a party, where crowd workers with different kinds of skills are needed to work collaboratively and finish the task. For micro-task assignment, there are many existing works and we refer the readers to \cite{DBLP:journals/pvldb/TongCS17} for more details. In this paper we focus on specialty-aware tasks assignment.

%第三段：现在对复杂任务的研究存在的问题是...因此我们怎么建模
Existing works \cite{TKDE16,WAIM16} on specialty-aware tasks assignment formulate that each crowd worker has multiple skills and will get a united fee if s/he is employed, which is not very practical as (1) workers often have unbalanced workloads, (2) workers may be confused what they should do in a task and (3) the payment and the workload do not often match. To solve the above drawbacks, in this paper we propose the \underline{S}pecialty-\underline{A}ware \underline{T}ask \underline{A}ssignment (SATA) problem where each crowd worker specify a fee for each of her/his skill to make the payment proportional to the workload.

%第四段：一个具体的例子
We then illustrate the STAT problem by a motivation example of organizing a party.

\begin{table}[t]
	\caption{Tasks and their lists of skills}
	\label{intro-ex1}
	\centering
	\begin{tabular}{|c|c|}
		\hline
		Tasks     & Lists of required skills \\
		\hline
		$t_1$     & $s_1$(music), $s_2$(drinks)                   \\
		\hline
		$t_2$     & $s_1$(music), $s_3$(barbecue), $s_4$(lights)         \\
		\hline
		$t_3$  	  & \tabincell{c}{$s_1$(music), $s_2$(drinks), $s_3$(barbecue),\\ $s_4$(lights), $s_5$(stage)} \\
		\hline	
	\end{tabular}
\end{table}

%\begin{table}[thp]\footnotesize
%	\caption{Tasks and their lists of skills}
%	\label{intro-ex1}
%	\centering
%	\begin{tabular*}{7.9cm}{@{\qquad}c@{\qquad}c@{\qquad}c@{\qquad}c@{\qquad}}
%		\toprule[0.75pt]
%		Tasks  & Lists of required skills    \\
%		\midrule[0.5pt]
%		$t_1$     & $s_1$(music), $s_2$(drinks)                   \\
%		$t_2$     & $s_1$(music), $s_3$(barbecue), $s_4$(lights)         \\
%		$t_3$  	  & \tabincell{c}{$s_1$(music), $s_2$(drinks), $s_3$(barbecue),\\ $s_4$(lights), $s_5$(stage)} \\
%		\bottomrule[0.75pt]
%		\multicolumn{2}{p{7.6cm}}{}
%	\end{tabular*}
%\end{table}

\begin{table}[t]
	\caption{Workers' skills and fees}
	\label{intro-ex2}
	\centering
	\begin{tabular}{|c|c|}
		\hline
		Workers  & Skills and fees \\
		\hline
		$w_1$     & $(s_1, 3), (s_2, 4), (s_4, 5)$                   \\
		\hline
		$w_2$     & $(s_3, 5), (s_5, 3)$                   \\
		\hline
		$w_3$     & $(s_4, 2)$                   \\
		\hline
		$w_4$     & $(s_1, 5), (s_5, 1)$         \\
		\hline
		$w_5$  	  & $(s_1, 2), (s_2, 2), (s_3, 3), (s_4, 6)$ \\
		\hline	
	\end{tabular}
\end{table}

%\begin{table}[thp]
%	\caption{}
%	\label{}
%	\centering
%	\begin{tabular*}{7.9cm}{@{\qquad}c@{\qquad}c@{\qquad}c@{\qquad}c@{\qquad}}
%		\toprule[0.75pt]
%		Workers  & Skills and fees  \\
%		\midrule[0.5pt]
%		$w_1$     & $(s_1, 3), (s_2, 4), (s_4, 5)$                   \\
%		$w_2$     & $(s_3, 5), (s_5, 3)$                   \\
%		$w_3$     & $(s_4, 2)$                   \\
%		$w_4$     & $(s_1, 5), (s_5, 1)$         \\
%		$w_5$  	  & $(s_1, 2), (s_2, 2), (s_3, 3), (s_4, 6)$ \\
%		\bottomrule[0.75pt]
%		\multicolumn{2}{p{7.6cm}}{}
%	\end{tabular*}
%\end{table}

\begin{figure}[htbp]
	% Requires \usepackage{graphicx}
	\centering
	\includegraphics[width=0.5\textwidth]{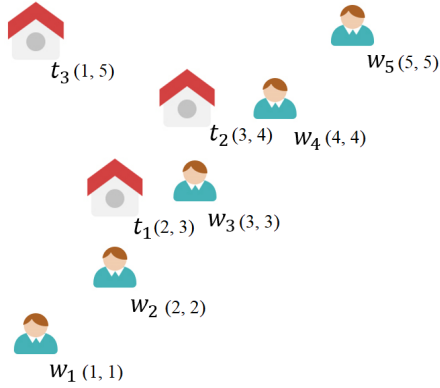}
	\label{fig:intro-ex}
	\caption{Locations of tasks and workers.}
	\label{fig:result1}
\end{figure}

\begin{example} \label{ex:fteam}
	Suppose we have three tasks of throwing parties, each has different styles and thus different kinds of works need to be done. For example, party 1 is a mini one and only needs music and drinks, while party 3 is ceremonious and requires music, drinks, barbecue, lights and a stage. The skill lists of the three tasks are shown in Table~\ref{intro-ex1}. Besides, we have some workers  shown in Table~\ref{intro-ex2}, each with different skills and corresponding fees. For example, if $w_1$ is required to finish the music job ($a_1$), s/he will be paid 3. Besides, each worker will get the transportation fee, which equals the distance from the worker to the assignment task times a global unit price. For example, Figure~1 shows the locations of tasks and workers, and if the global unit price is 0.5, the transportation fees for assigning $w_1$ to $t_1$ is $\sqrt{5}\times0.5\approx1.12$.
\end{example}

Motivated by the example above, we will formalize the STAT problem, which aims to efficiently assign crowd workers to specialty-aware tasks to maximize the total utility of the assignment. Note that existing works either focus on assigning workers to micro tasks to optimize different goals, or assume that the 
workers have a united fee. Thus, their methods cannot be directed adopted to solve our problem.

In this paper, we first prove the NP-hardness of the SATA problem, indicating that SATA is not tractable and it is challenging to gain the optimal solution. Therefore, we propose two efficient and effective heuristics to solve it. 

To summarize, we make the following contributions.

\begin{itemize}
	\item We formally define a new task assignment problem in spatial crowdsourcing, called the \underline{S}pecialty-\underline{A}ware \underline{T}ask \underline{A}ssignment (SATA) problem. 
	\item We prove the SATA problem is NP-hard, and develop two efficient heuristics to solve it.
	\item We verify the effectiveness and efficiency of the proposed methods with extensive experiments on real and synthetic datasets.
\end{itemize}

The rest of the paper is organized as follows. We define the SATA problem and prove its NP-hardness in Section 2. Section 3 discusses extensive experiment results on both synthetic and real datasets. We review related works in Section 4 and conclude in Section 5.

\section{Problem Definition}

We first introduce two basic concepts, namely Task and Worker. Then, we introduce how to calculate the reward of worker. Finally, we formally give the definition of the \underline{S}pecialty-\underline{A}ware \underline{T}ask \underline{A}ssignment (SATA) problem.

\begin{definition} [Worker]
	A worker $w$ is defined as $<L_w, S_w, P_w>$, where $L_w$ is the location of $w$ which can be described by longitude and latitude, $S_w = <s^w_1, s^w_2, \cdots, s^w_{|S_w|}>$ is the list of skills that $w$ masters, and $P_w = <p^w_1, p^w_2, \cdots, p^w_{|S_w|}>$ is the list of fees for each skill in $S_w$.
\end{definition}

Similar to the definition of a worker, a task is formally defined as follows.

\begin{definition} [Task]
	A task $t$ is defined as $<L_t, S_t, B_t>$, where $L_t$ is the location of $t$ which can be described by longitude and latitude, $S_t = <s^t_1, s^t_2, \cdots, s^t_{|S_t|}>$ is the list of skills that are needed to complete $t$ collaboratively, and $B_t$ is the total monetary budget of $t$.
\end{definition}

Briefly, a worker's reward includes two parts: (1) transportation fee, which is directly proportional to the distance between the worker and the task; (2) labor fee, which is the sum of the fees for the skills used to perform a task.

\begin{definition} [Reward of Worker]
	The reward of task $w$ to perform task $t$ equals $r_w = \gamma \cdot dis(L_w, L_t) + \sum_{s\in S'_w}{p^w_s}$, where $dis(L_w, L_t)$ is the distance between $L_w$ and $L_t$, which can be Euclidean distance or road network distance, $\gamma$ is a global parameter representing the unit transportation fee, and $S'_w$ is the set of skills that $w$ uses to perform the task.
\end{definition}

We define the utility of a task as follows.

\begin{definition} [Utility of Task]
	The utility of task $t$ is defined as $u_t = B_t - \sum_{t\in W_t}r_t$, where $B_t$ is the budget of the task and $\sum_{t\in W_t}r_t$ is the summation of rewards of workers assigned to $t$ if $t$ is completed. If $t$ cannot be finished, the utility is zero.
\end{definition}

We finally define our problem as follows.

\begin{definition} [\underline{S}pecialty-\underline{A}ware \underline{T}ask \underline{A}ssignment (SATA) Problem]
	Given a set of tasks $T$, a set of workers $W$ and a global unit transportation fee $ \gamma $, the problem is to assign workers to tasks to maximize the total utility of the completed tasks and the following constraints should be satisfied:
	\begin{itemize}
		\item \textbf{Specialty Constraint}: a task can be completed as long as the workers assigned to it can cover the required skills of the task;
		\item \textbf{Budget Constraint}: the total rewards of workers assigned to a task cannot exceed the task's total budget;
	\end{itemize}
\end{definition}

We then prove the hardness of SATA problem.

\begin{theorem}
	The SATA problem is NP-hard.
\end{theorem}

\begin{proof}
	We prove through a reduction from the set cover problem  
	\cite{ApproALG}
	
	We first introduce the set cover problem. Given a universe $U = \{s_1, s_2, \cdots, s_n\}$ and its $m$ subsets $S_1, S_2, \cdots, S_m \subseteq U$, $\cup_{i=1}^m{S_i}=U$. Each $S_i$ is associated with a cost $c_i$. The set cover problem is to find a set $K \subseteq \{1, 2, \cdots, m\}$ to minimize $\sum_{i\in K}c_i$ satisfying $\cup_{i\in K}S_i=U$.
	
	We next show how to transform the set cover problem to an instance of our SATA problem. We only have one task $t$ which requires skills $S_t = U$ and has infinite budget $B_t$. For $m$ workers $\{w_1, w_2, \cdots, w_m\}$, their required fees for skills are all zero, and we adjust their locations and $\gamma$ to make their transportation fee to perform $t$ is $c_i$. For this instance of our SATA problem, we aim to find a set of workers $K$ to maximize the utility of $t$, which equals to minimize $\sum_{i\in K}c_i$. In this way, we reduce set cover problem to our SATA problem. As the set cove problem is known to be NP-hard \cite{ApproALG}, SATA problem is also NP-hard.
\end{proof}

\section{Algorithms}
In this section, we give two efficient heuristic algorithms to solve the SATA problem. 

\subsection{Total Budget Based Algorithm}
Our first algorithm is called the \underline{T}otal \underline{B}udget Based \underline{A}lgorithm (TBA). The main idea is that we always try to assign workers to the tasks with the largest budget. During the procedure of task assignment, we refer to the greedy algorithm to solve the set cover problem \cite{ApproALG}.

The procedure of TBA is shown in algorithm-\ref{alg:TBA}. The algorithm takes the set of workers $W$ and set of tasks $T$ as input, and return an assignment $M$ between them as shown in lines 1-2. In line 3, the algorithm first sorts the tasks in $T$ in descending order according to their budgets, and the sorted result is saved in $Q$. In lines 4-13, for each task $t$ in $Q$, we refer to the greedy algorithm to solve the set cover problem \cite{ApproALG} to assign workers. Specifically, in lines 5, we find worker with minimum $\frac{r_w}{|S'_w\cap S_t|}$. Notes that here $S'w$ considers all possible subsets of $S_w$. In lines 6-7, we update $M$, $W$ and $S_t$. In lines 9-11, if $S_t$ is $\emptyset$, which means it can be completed, we break the loop and start to assign workers for the next task.

\begin{algorithm}[t]
	\SetKwInOut{Input}{input}\SetKwInOut{Output}{output}
	\Input{set of workers $W$, set of tasks $T$}
	\Output{Assignment $M$}
	$Q \leftarrow$ sorting tasks in $T$ according to their budgets in descending order\;
	\ForEach{$t$ in $Q$}{
		Assign $w \in W$ to $t$ with minimum $\frac{r_w}{|S'_w\cap S_t|}$\;
		Update $M$ and $W$\;
		$S_t \leftarrow S_t - S'_w$\;
		\If{$S_t$ is $\emptyset$}{
			Break\;
		}
		
	}
	\Return{$M$} 
	\caption{\underline{T}otal \underline{B}udget Based \underline{A}lgorithm (TBA)}
	\label{alg:TBA}
\end{algorithm}

%\begin{algorithm}\small
%	\caption{\underline{T}otal \underline{B}udget Based \underline{A}lgorithm (TBA)} \centering
%	\begin{algorithmic}[1]
%		\STATE{Input: set of workers $W$, set of tasks $T$}
%		\STATE{Output: Assignment $M$}
%		\STATE{$Q \leftarrow$ sorting tasks in $T$ according to their budgets in descending order}
%		\FOR{$t$ in $Q$}
%		\LOOP
%		\STATE{Assign $w \in W$ to $t$ with minimum $\frac{r_w}{|S'_w\cap S_t|}$}
%		\STATE{Update $M$ and $W$}
%		\STATE{$S_t \leftarrow S_t - S'_w$}
%		\IF{$S_t$ is $\emptyset$}
%		\STATE{Break}
%		\ENDIF
%		\ENDLOOP
%		\ENDFOR
%		\STATE{Return $M$}
%	\end{algorithmic}
%	\label{alg:TBA}
%\end{algorithm}	

\begin{example}
	Back to our running example in Example 1. TBA first finds the task with the largest total budget, which is $t_3$. The it starts to assign workers for $t_3$. As $w_3$ has the minimal $\frac{r_w}{|S'_w\cap S_t|}$ of 2, we first assign $w_3$ to $t_3$. After assigning $w_3$, $t_3$'s list of skills has not been covered, thus we assign $w_5$ to $t_3$ with $\frac{r_w}{|S'_w\cap S_t|}$ of $\frac{7}{3}$. We finally assign $w_4$ to $t_3$ and the total reward paid to $w_3$, $w_4$ and $w_5$ is $2+2+2+3+1+(\sqrt{2}+4+\sqrt{10})\times0.5\approx15$. Thus, the utility of $t_3$ is $30-15=15$. Similarly, we  assign workers to $t_2$ and $t_3$ successively, and the final utility of TBA is 21.08.
\end{example}

\textbf{Complexity.} If we take the maximum number of skills a worker may have as a constant, the time complexity of TBA is $O(max\{|T|log|T|, |T||W|\})$.

\subsection{Average Budget Based Algorithm}
The TBA algorithms only considers the total budget of tasks. However, a large budget may result from a large number of skills required in the task. Thus, in this subsection, we propose another algorithm, called \underline{A}verage \underline{B}udget Based \underline{A}lgorithm (ABA). The main idea is that we first measure the average budget of all the tasks, and prefer to assign workers to tasks with a larger average budget.

\begin{algorithm}[t]
	\SetKwInOut{Input}{input}\SetKwInOut{Output}{output}
	\Input{set of workers $W$, set of tasks $T$}
	\Output{Assignment $M$}
	$Q \leftarrow$ sorting tasks in $T$ according to their budgets in descending order\;
	\ForEach{$t$ in $Q$}{
		Assign $w \in W$ to $t$ with minimum $\frac{r_w}{|S'_w\cap S_t|}$\;
		Update $M$ and $W$\;
		$S_t \leftarrow S_t - S'_w$\;
		\If{$S_t$ is $\emptyset$}{
			Break\;
		}
	}
	\Return{$M$} 
	\caption{\underline{A}verage \underline{B}udget Based \underline{A}lgorithm (ABA)}
	\label{alg:ABA}
\end{algorithm}

%\begin{algorithm}\small
%	\caption{\underline{T}otal \underline{B}udget Based \underline{A}lgorithm (TBA)} \centering
%	\begin{algorithmic}[1]
%		\STATE{Input: set of workers $W$, set of tasks $T$}
%		\STATE{Output: Assignment $M$}
%		\STATE{$Q \leftarrow$ sorting tasks in $T$ according to $\frac{B_t}{|S_t|}$ in descending order}
%		\FOR{$t$ in $Q$}
%		\LOOP
%		\STATE{Assign $w \in W$ to $t$ with minimum $\frac{r_w}{|S'_w\cap S_t|}$}
%		\STATE{Update $M$ and $W$}
%		\STATE{$S_t \leftarrow S_t - S'_w$}
%		\IF{$S_t$ is $\emptyset$}
%		\STATE{Break}
%		\ENDIF
%		\ENDLOOP
%		\ENDFOR
%		\STATE{Return $M$}
%		\end{algorithmic}
%	\label{alg:ABA}
%\end{algorithm}	

The pseudo codes of ABA is shown in Algorithm-2. The biggest difference between TBA and ABA lies on line 3. In TBA, we first sort tasks in $T$ based on average budget, which is defined as $\frac{B_t}{|S_t|}$. The procedure of how to assign workers to a given task is the same as TBA, which is shown in lines 4-13.

\begin{example}
	Back to our running example in Example 1. Different from TBA, ABA first finds the task with the largest average budget, which is $t_1$. Then it assigns workers to $t_1$. As $w_5$ has the minimal $\frac{r_w}{|S'_w\cap S_t|}$ of 2, we first assign $w_5$ to $t_1$. After assigning $w_5$, we find $t_1$'s list of skills has been covered, thus the total utility is $20-(2+2)-(\sqrt{13}\times 0.5\approx14.20)$. Similarly, we next assign workers for $t_2$ and $t_3$, and the final utility of TBA is 25.78.
\end{example}

\textbf{Complexity.} If we take the maximum number of skills a worker may have as a constant, the time complexity of ABA is also $O(max\{|T|log|T|, |T||W|\})$.
\section{Evaluation}

\subsection{Experiment Setup}
We use real and synthetic datasets to evaluate our algorithms. Real data comes from CSTO (http://www.csto.com/), which is an outsource task platform. In the CSTO dataset, each task is associated with a set of skills needed to complete the software development task, and each coder is associated with a set of skills and an average price which can be deduced from the history data. Since the CSTO data is not associated with location information, we generate the distance of each coder from the task following uniform distribution. For synthetic data, based on the observation from real data set, the price of skills owned by a worker and the budget of a task follow Gaussian distribution, respectively. Statistics of the synthetic data are shown in Table3, where we mark our default settings in bold font.

\begin{table}[t]
	\caption{Synthetic Dataset}
	\label{Table:1}
	\centering
	\begin{tabular}{|c|c|}
		\hline
		Factor  & Setting    \\
		\hline
		$|T|$     & 100 300 \textbf{500} 700 900                          \\
		\hline
		$|W|$     & 1000 3000 \textbf{5000} 7000 9000                     \\
		\hline
		$\gamma$  & 0.1 0.3 \textbf{0.5} 0.7 0.9                          \\
		\hline
		$B_t$     & 60 80 \textbf{100} 120 140                         \\
		\hline
		$P_w$     & 10 15 \textbf{20} 25 30                          \\
		\hline
		$|S|$	  & 10 20 \textbf{30} 40 50								\\
		\hline	
	\end{tabular}
\end{table}

%\begin{table}[thp]\footnotesize
%	\caption{}\label{}
%	\centering
%	\begin{tabular*}{7.9cm}{@{\qquad}c@{\qquad}c@{\qquad}c@{\qquad}c@{\qquad}}
%		\toprule[0.75pt]
%		Factor  & Setting    \\
%		\midrule[0.5pt]
%		$|T|$     & 100 300 \textbf{500} 700 900                          \\
%		$|W|$     & 1000 3000 \textbf{5000} 7000 9000                     \\
%		$\gamma$  & 0.1 0.3 \textbf{0.5} 0.7 0.9                          \\
%		$B_t$     & 60 80 \textbf{100} 120 140                         \\
%		$P_w$     & 10 15 \textbf{20} 25 30                          \\
%		$|S|$	  & 10 20 \textbf{30} 40 50								\\
%		\bottomrule[0.75pt]
%		\multicolumn{2}{p{7.6cm}}{}
%	\end{tabular*}
%\end{table}

\begin{figure*}[htb]
	% Requires \usepackage{graphicx}
	\centering
	\subfloat[\scriptsize{Cardinality of varying $|T|$}]{
		%\begin{minipage}{5cm}
		\centering
		\includegraphics[scale=0.3]{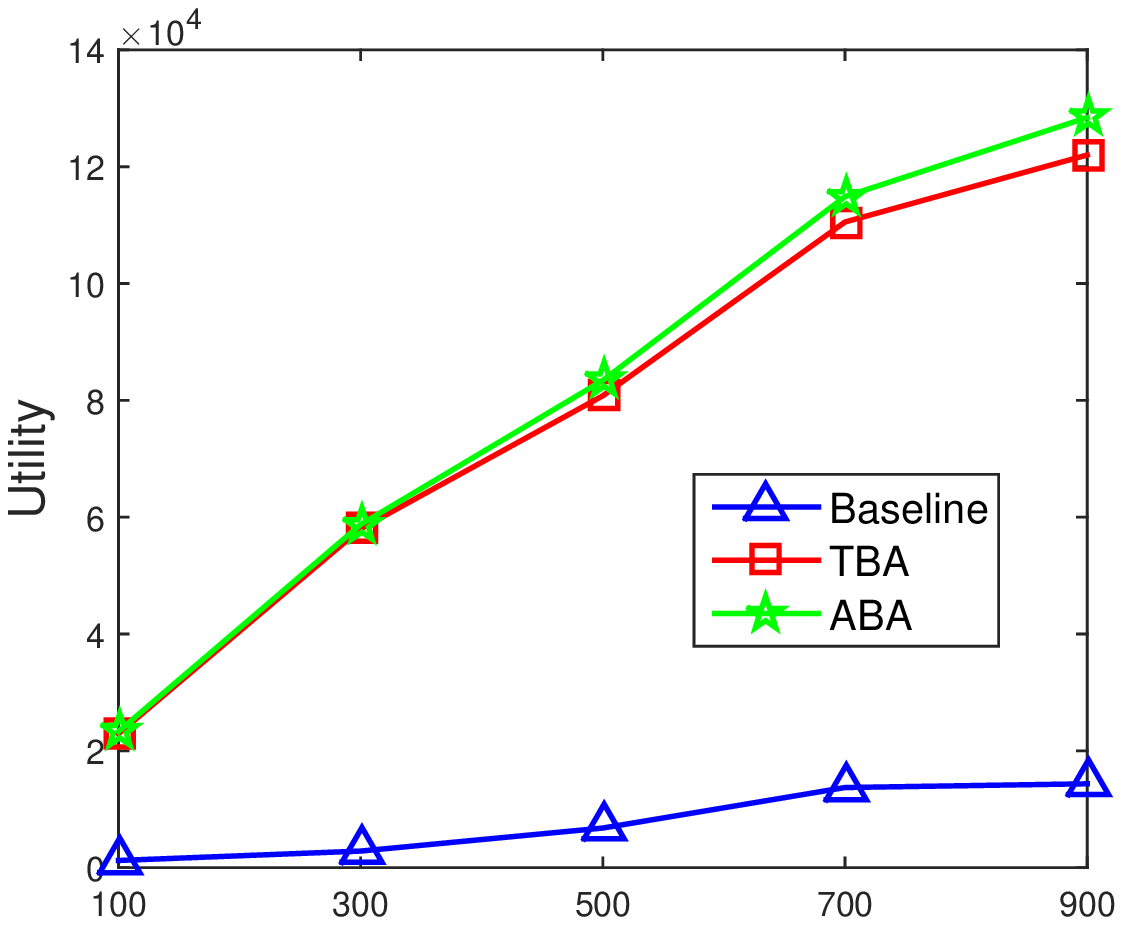}
		\label{fig:result1a}
		%\end{minipage}
	}
	\subfloat[\scriptsize{Running Time of varying $|T|$}]{
		%\begin{minipage}{5cm}
		\centering
		\includegraphics[scale=0.3]{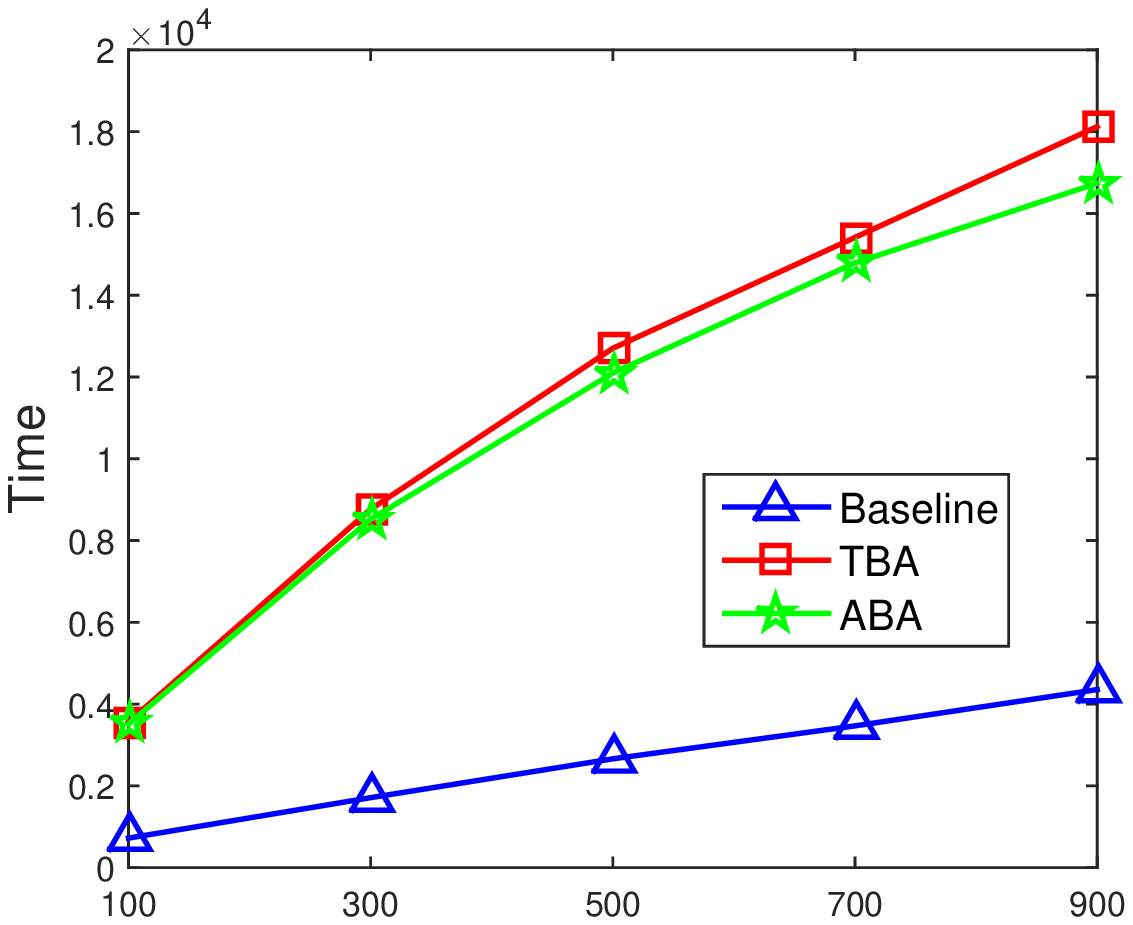}
		\label{fig:result1b}
		%\end{minipage}
	}
	\subfloat[\scriptsize{Memory of varying $|T|$}]{
		%\begin{minipage}{5cm}
		\centering
		\includegraphics[scale=0.3]{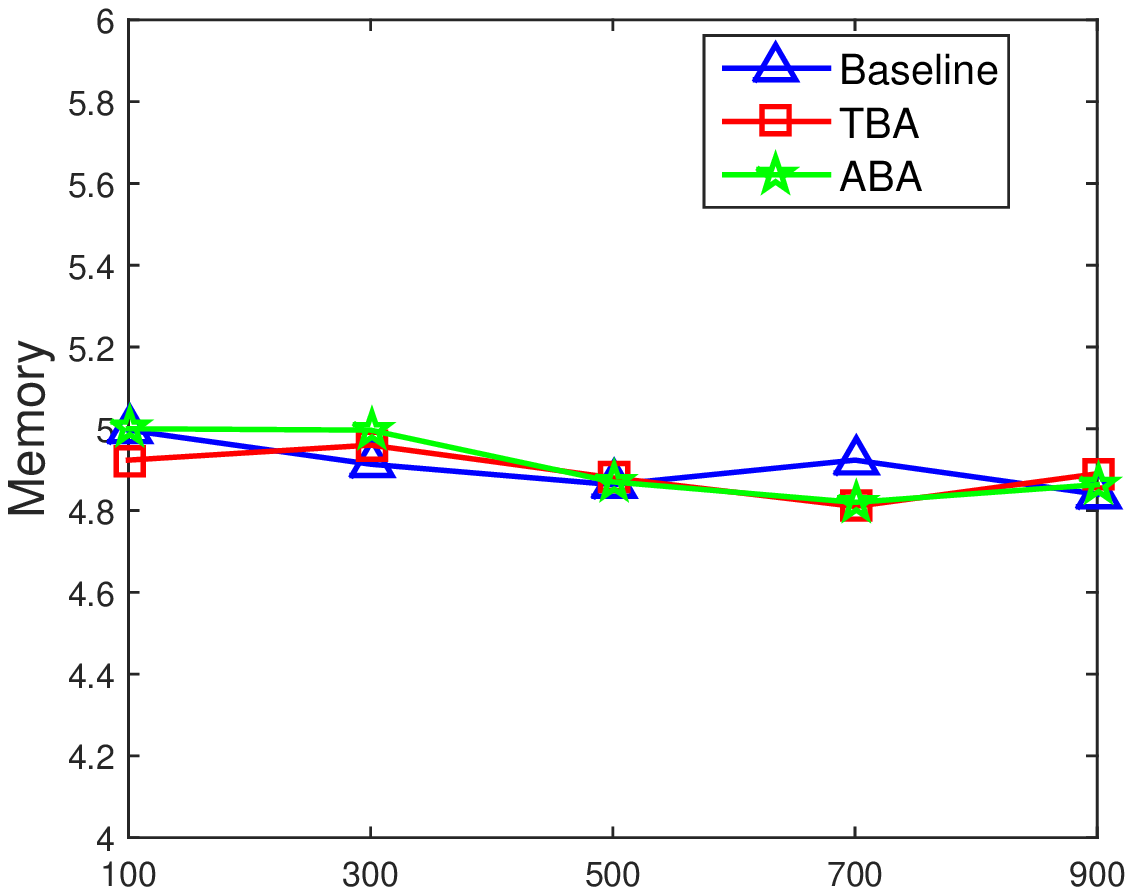}
		\label{fig:result1c}
		%\end{minipage}
	}
	\caption{Results on varying $|T|$.}
\end{figure*}

\begin{figure*}[htb]
	\centering
	\subfloat[\scriptsize{Cardinality of varying $|W|$}]{
		%\begin{minipage}{5cm}
		\centering
		\includegraphics[scale=0.3]{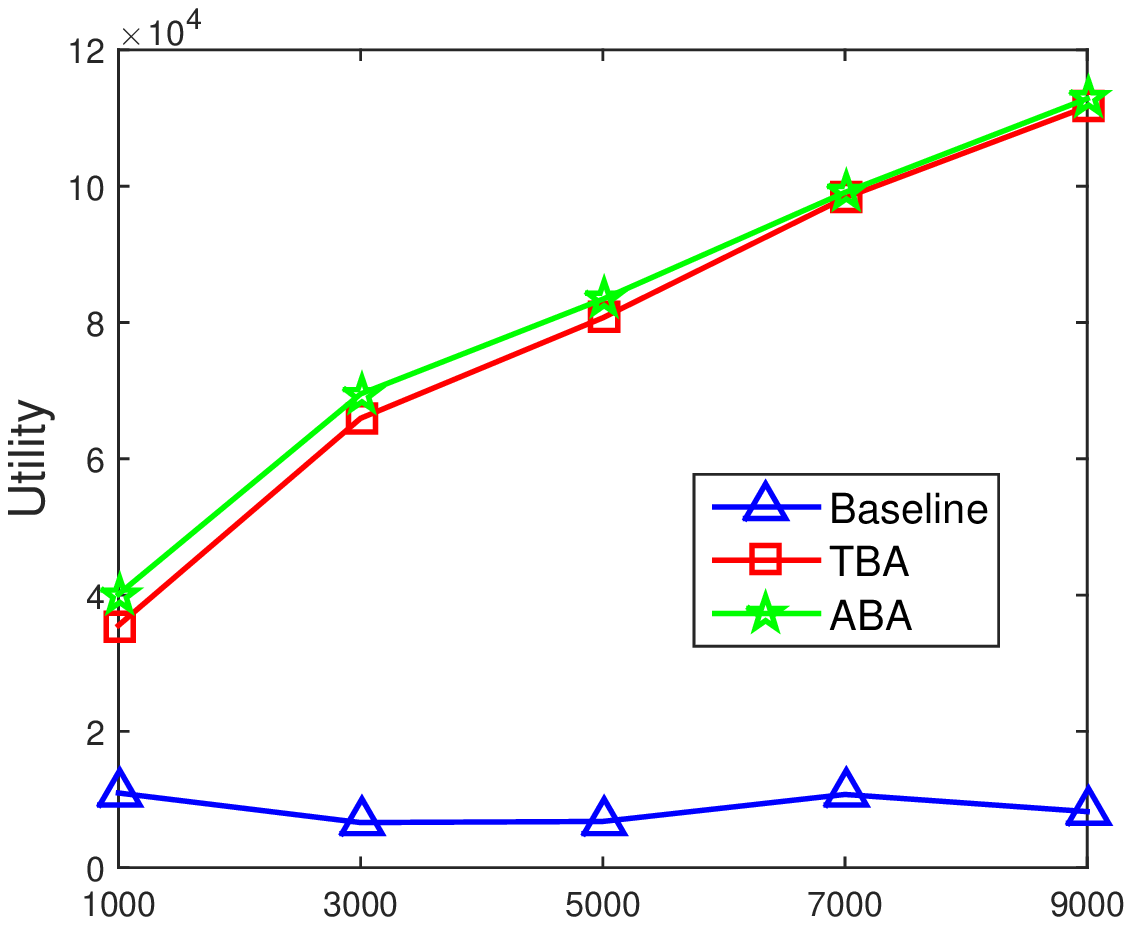}
		\label{fig:result1d}
		%\end{minipage}
	}
	\subfloat[\scriptsize{Running Time of varying $|W|$}]{
		%\begin{minipage}{5cm}
		\centering
		\includegraphics[scale=0.3]{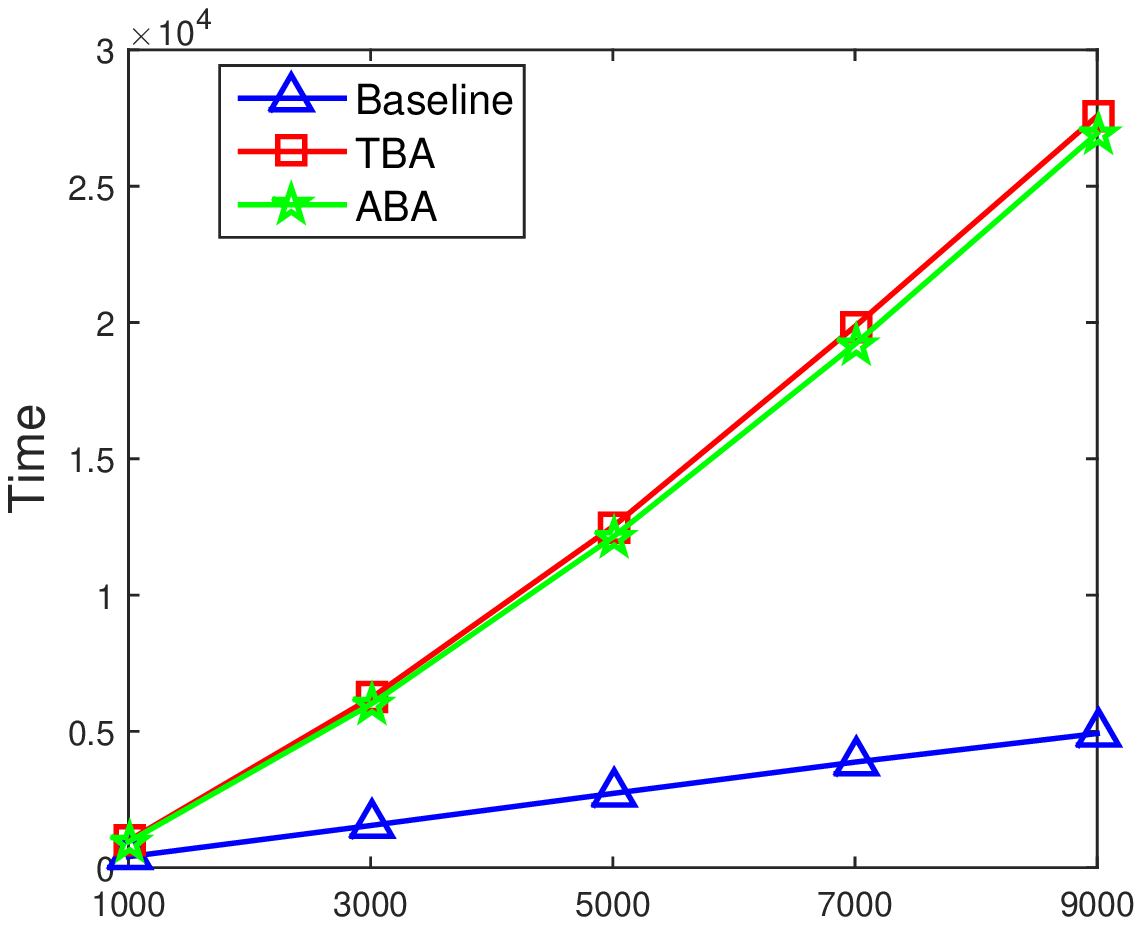}
		\label{fig:result1e}
		%\end{minipage}
	}
	\subfloat[\scriptsize{Memory of varying $|W|$}]{
		%\begin{minipage}{5cm}
		\centering
		\includegraphics[scale=0.3]{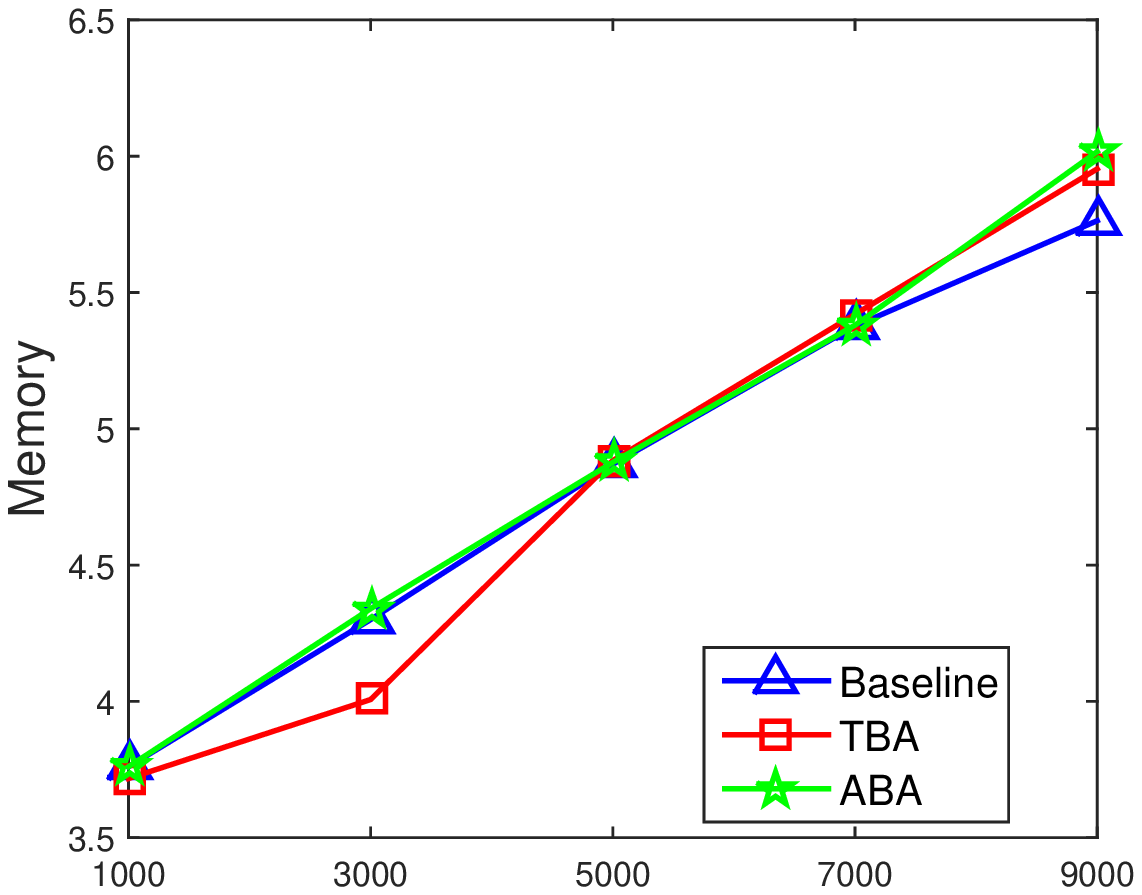}
		\label{fig:result1f}
		%\end{minipage}
	}
	\caption{Results on varying $|W|$.}
\end{figure*}

\begin{figure*}[htb]
	\centering
	\subfloat[\scriptsize{Cardinality of varying $\gamma$}]{
		%\begin{minipage}{5cm}
		\centering
		\includegraphics[scale=0.3]{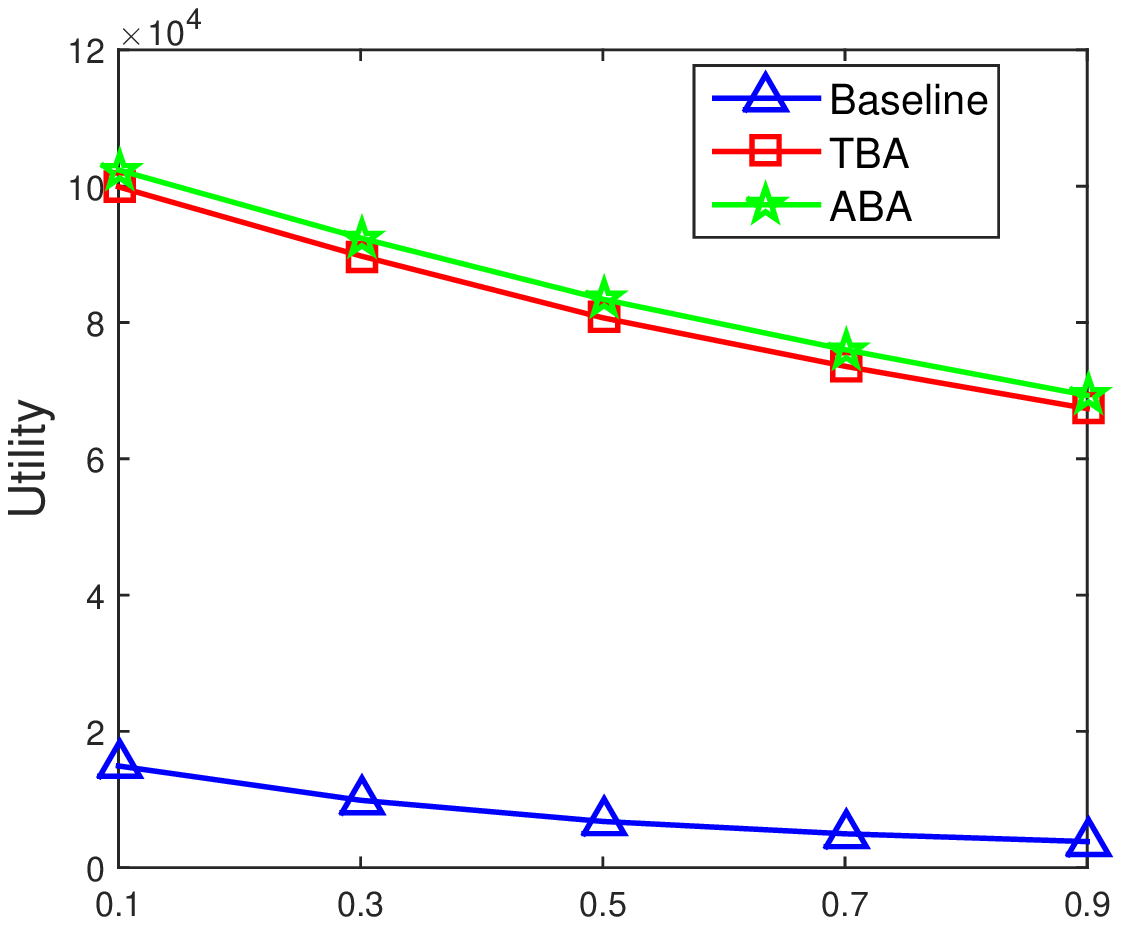}
		\label{fig:result1g}
		%\end{minipage}
	}
	\subfloat[\scriptsize{Running Time of varying $\gamma$}]{
		%\begin{minipage}{5cm}
		\centering
		\includegraphics[scale=0.3]{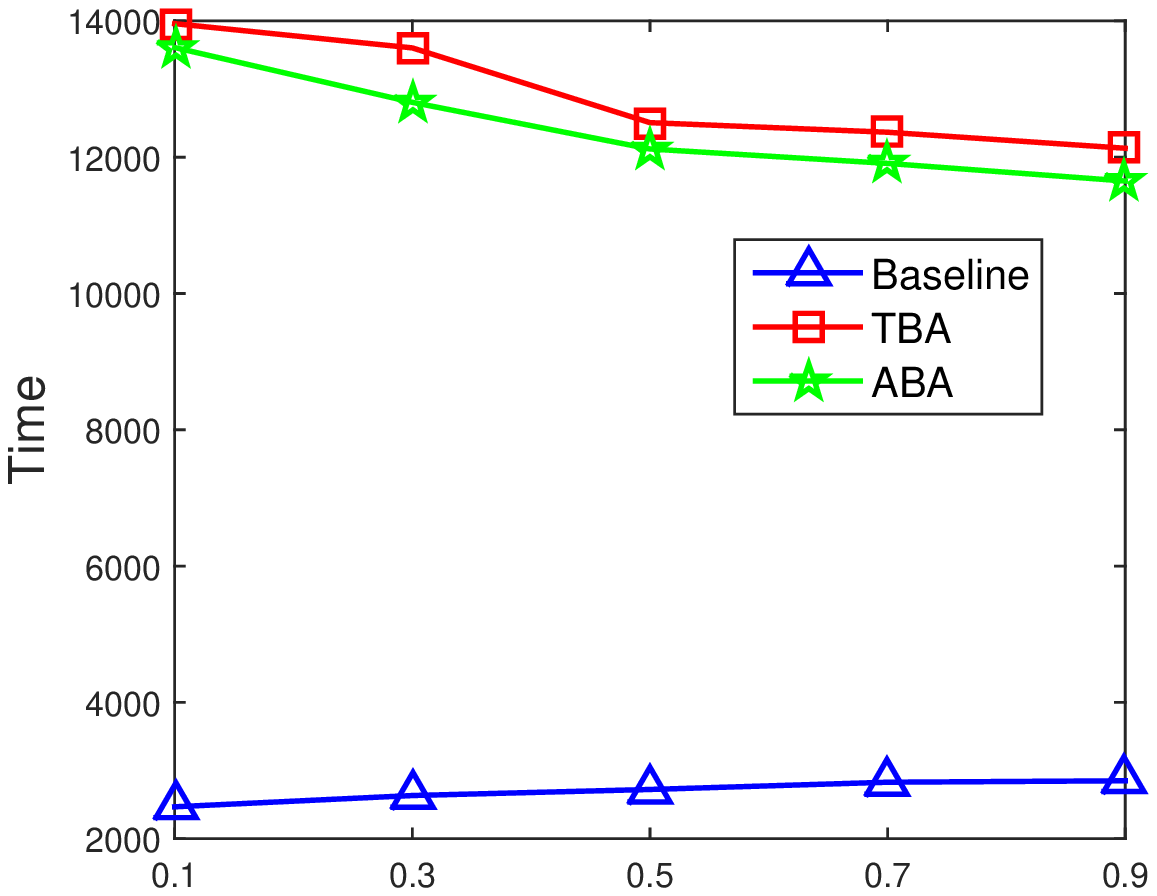}
		\label{fig:result1h}
		%\end{minipage}
	}
	\subfloat[\scriptsize{Memory of varying $\gamma$}]{
		%\begin{minipage}{5cm}
		\centering
		\includegraphics[scale=0.3]{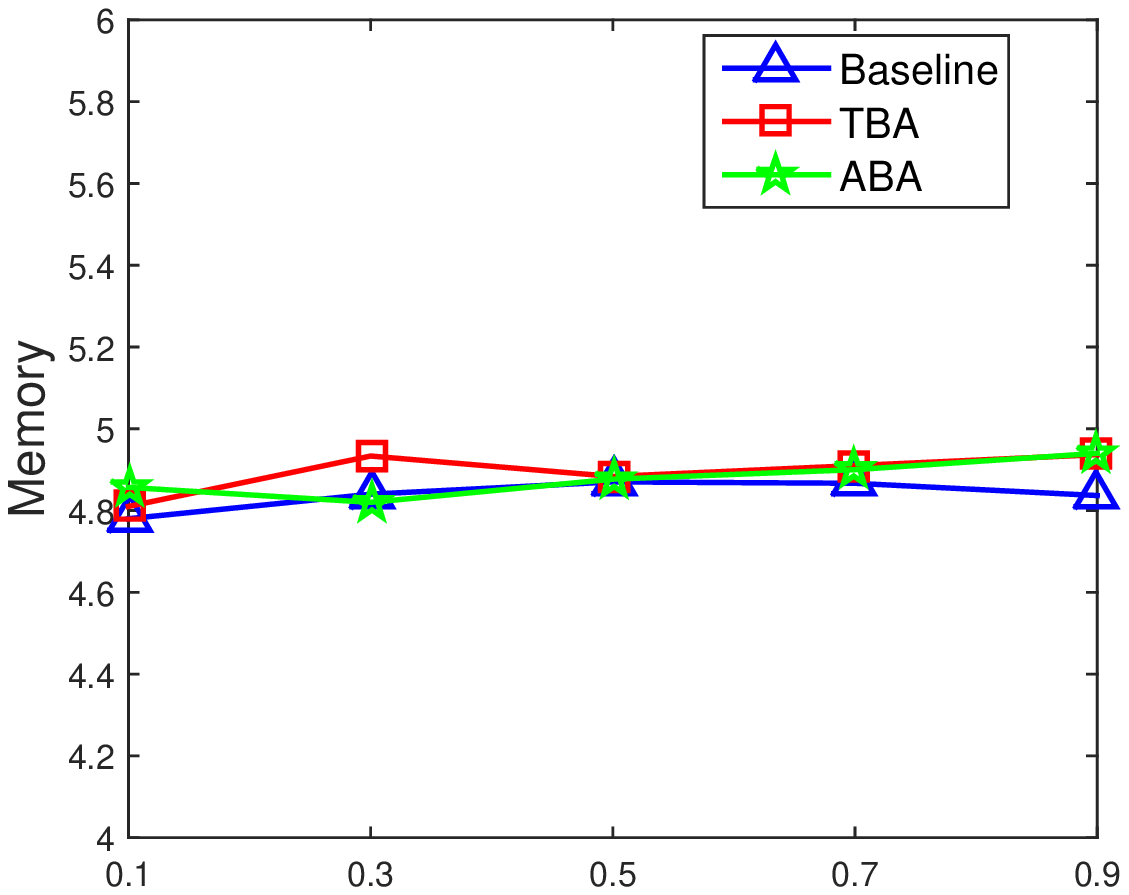}
		\label{fig:result1i}
		%\end{minipage}
	}
	\caption{Results on varying $\gamma$.}
\end{figure*}

\subsection{Experiment Results}

In this subsection, we test the performance of our proposed algorithms by setting different parameters. We evaluate two exact algorithm, called TBA and ABA, and a baseline algorithm in terms of total utility score, running time and memory cost, and study the effect of varying parameters on the performance of the algorithms. The baseline algorithm uses a simple random strategy, which assigns workers to tasks randomly. The algorithms are implemented in CodeBlocks16.1, and the experiments are performed on a machine with Intel(R) Core(TM) i5 2.50GHZ CPU and 8GB main memory.

%\subfloat[\scriptsize{\scriptsize{Utility of Varying $|T|$}]{
%	\includegraphics[scale=0.3]{figure/experiments/taskn_u.eps}
%	\label{fig:exp_tn_u}
%}

\begin{figure*}[htb]
	% Requires \usepackage{graphicx}
	\centering
	\subfloat[\scriptsize{Cardinality of varying $B_t$}]{
		%\begin{minipage}{5cm}
		\centering
		\includegraphics[scale=0.3]{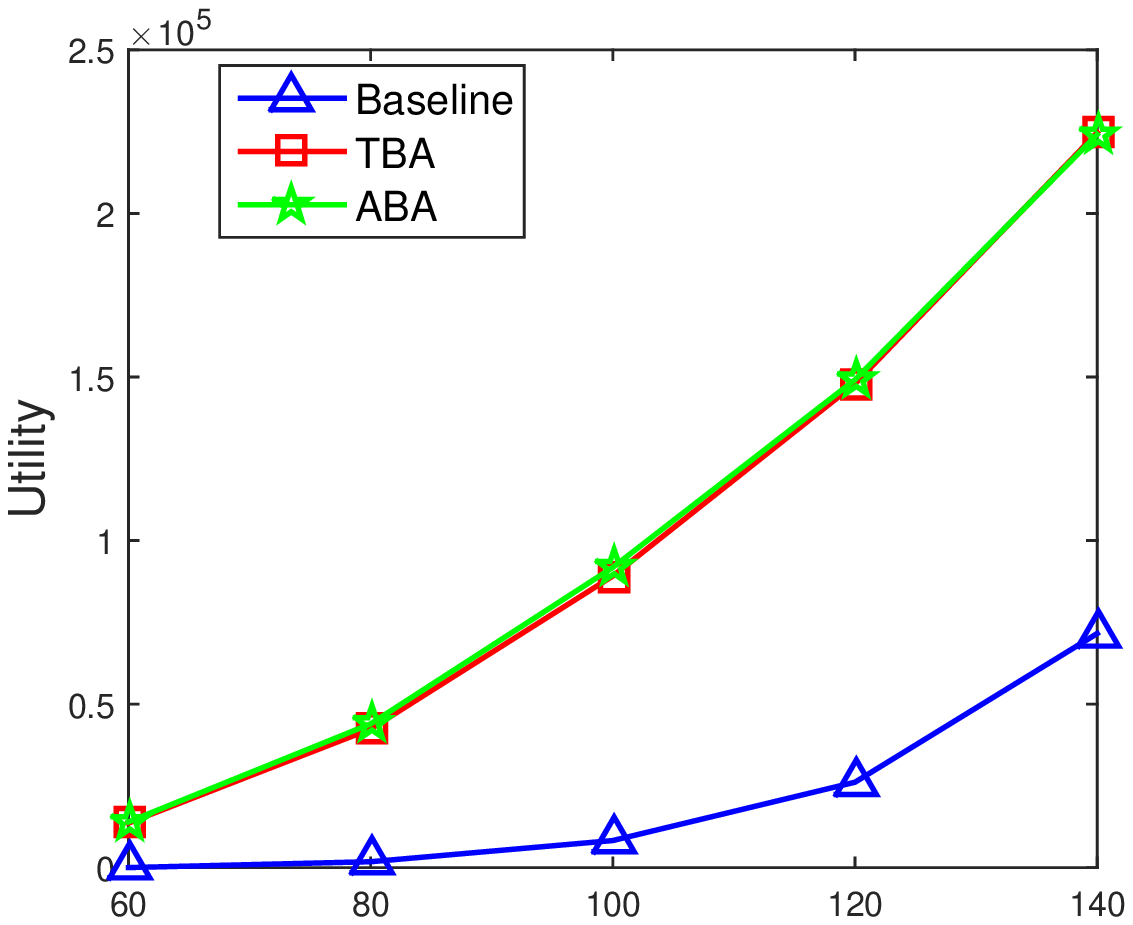}
		\label{fig:result2a}
		%\end{minipage}
	}
	\subfloat[\scriptsize{Runnint Time of varying $B_t$}]{
		%\begin{minipage}{5cm}
		\centering
		\includegraphics[scale=0.3]{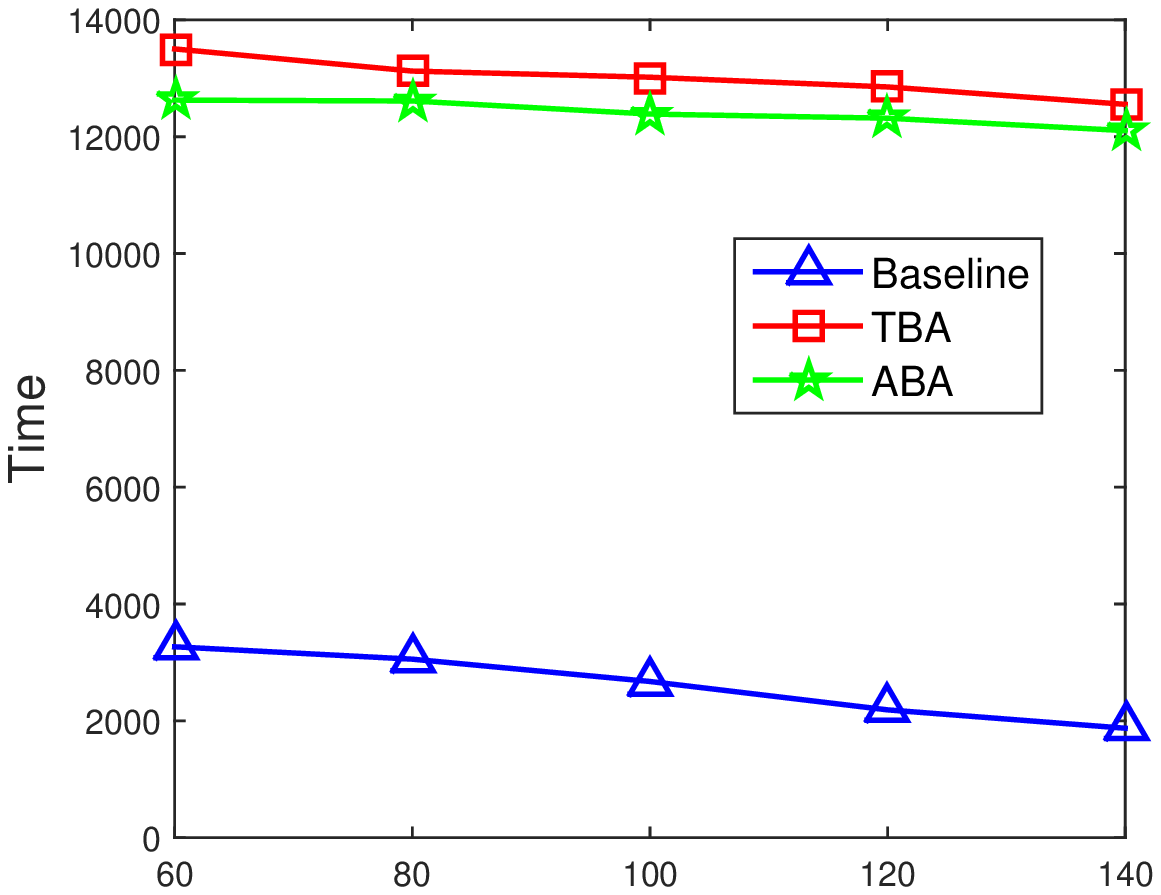}
		\label{fig:result2b}
		%\end{minipage}
	}
	\subfloat[\scriptsize{Memory of varying $B_t$}]{
		%\begin{minipage}{5cm}
		\centering
		\includegraphics[scale=0.3]{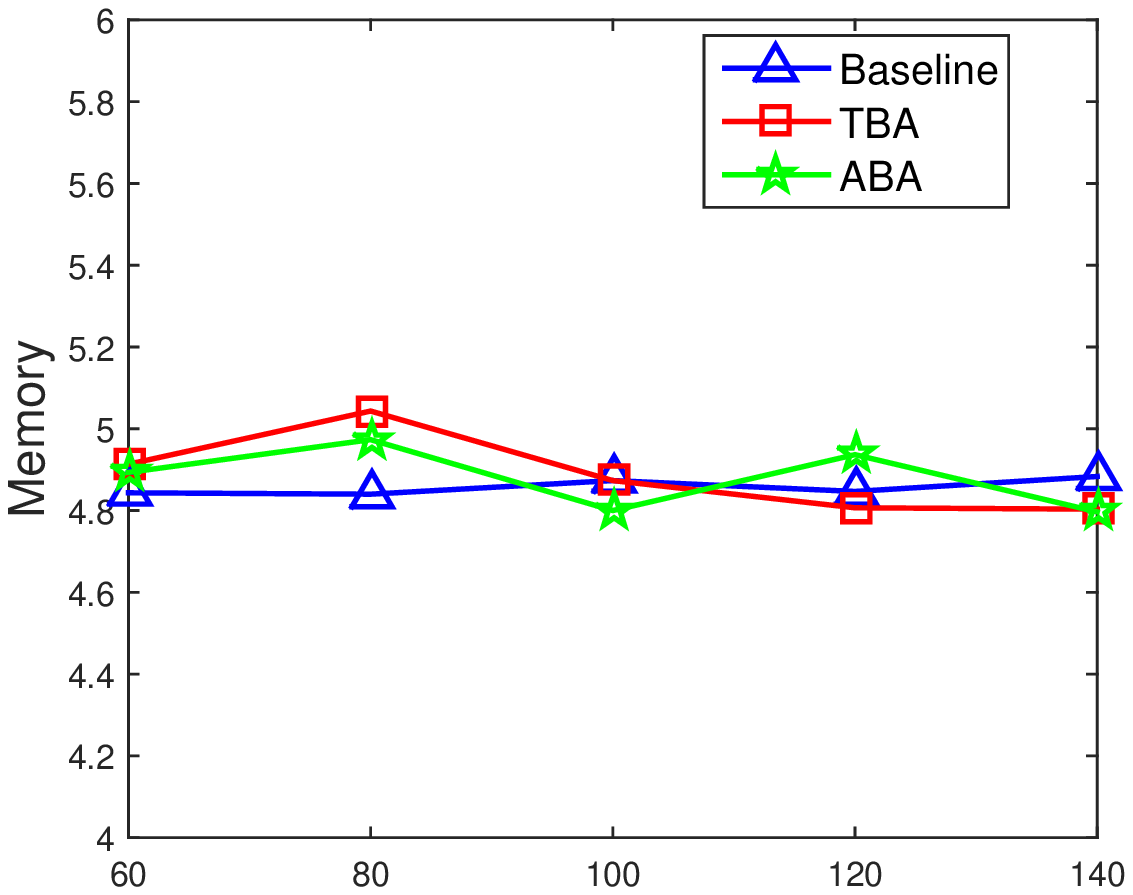}
		\label{fig:result2c}
		%\end{minipage}
	}
	\caption{Results on varying $B_t$.}
\end{figure*}

\begin{figure*}[htb]
	\centering
	\subfloat[\scriptsize{Cardinality of varying $P_w$}]{
		%\begin{minipage}{5cm}
		\centering
		\includegraphics[scale=0.3]{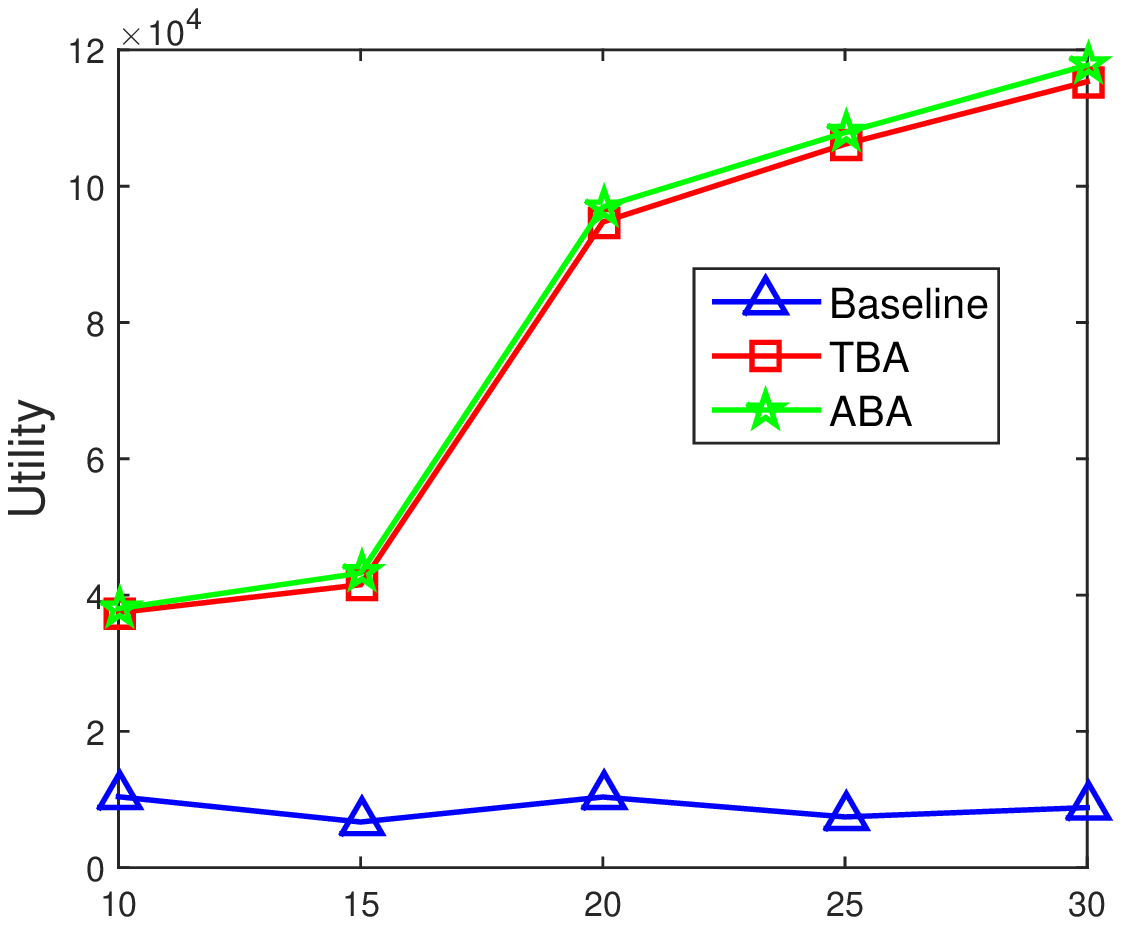}
		\label{fig:result2d}
		%\end{minipage}
	}
	\subfloat[\scriptsize{Running Time of varying $P_w$}]{
		%\begin{minipage}{5cm}
		\centering
		\includegraphics[scale=0.3]{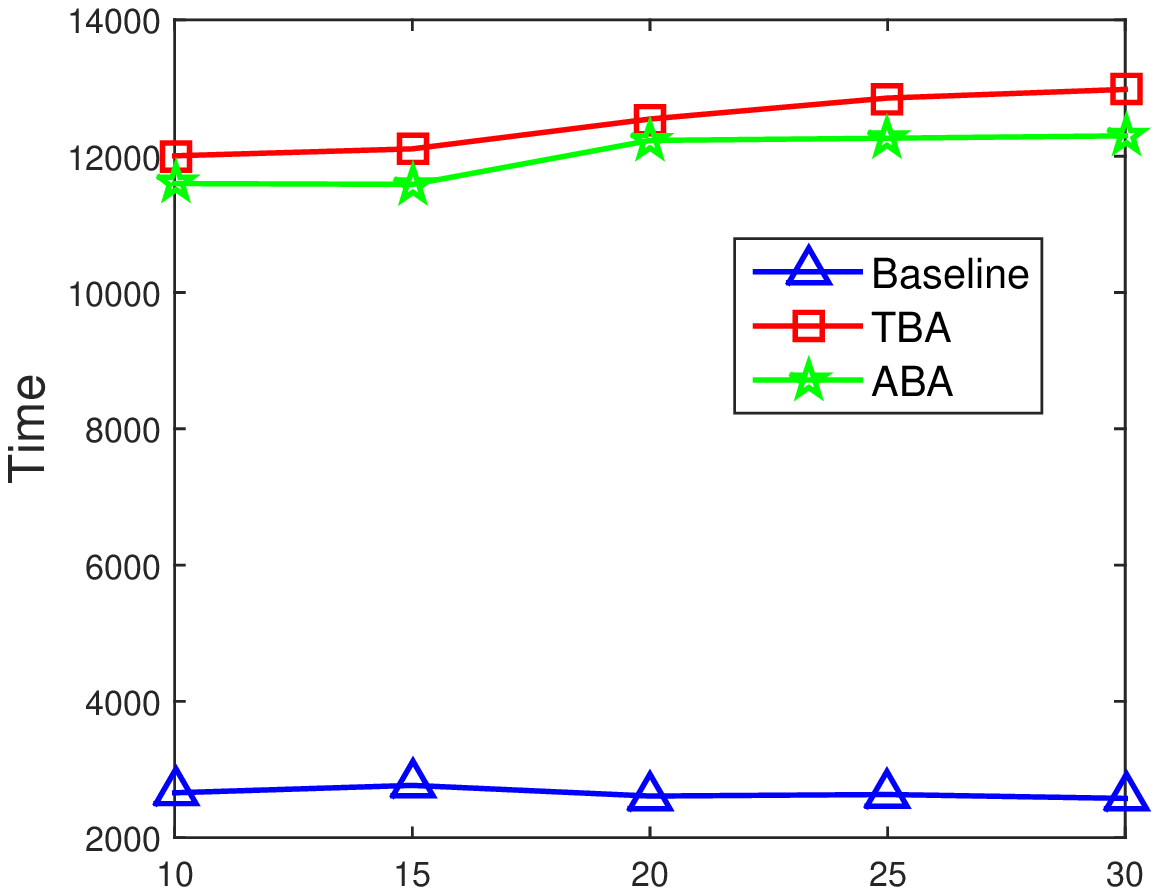}
		\label{fig:result2e}
		%\end{minipage}
	}
	\subfloat[\scriptsize{Memory of varying $P_w$}]{
		%\begin{minipage}{5cm}
		\centering
		\includegraphics[scale=0.3]{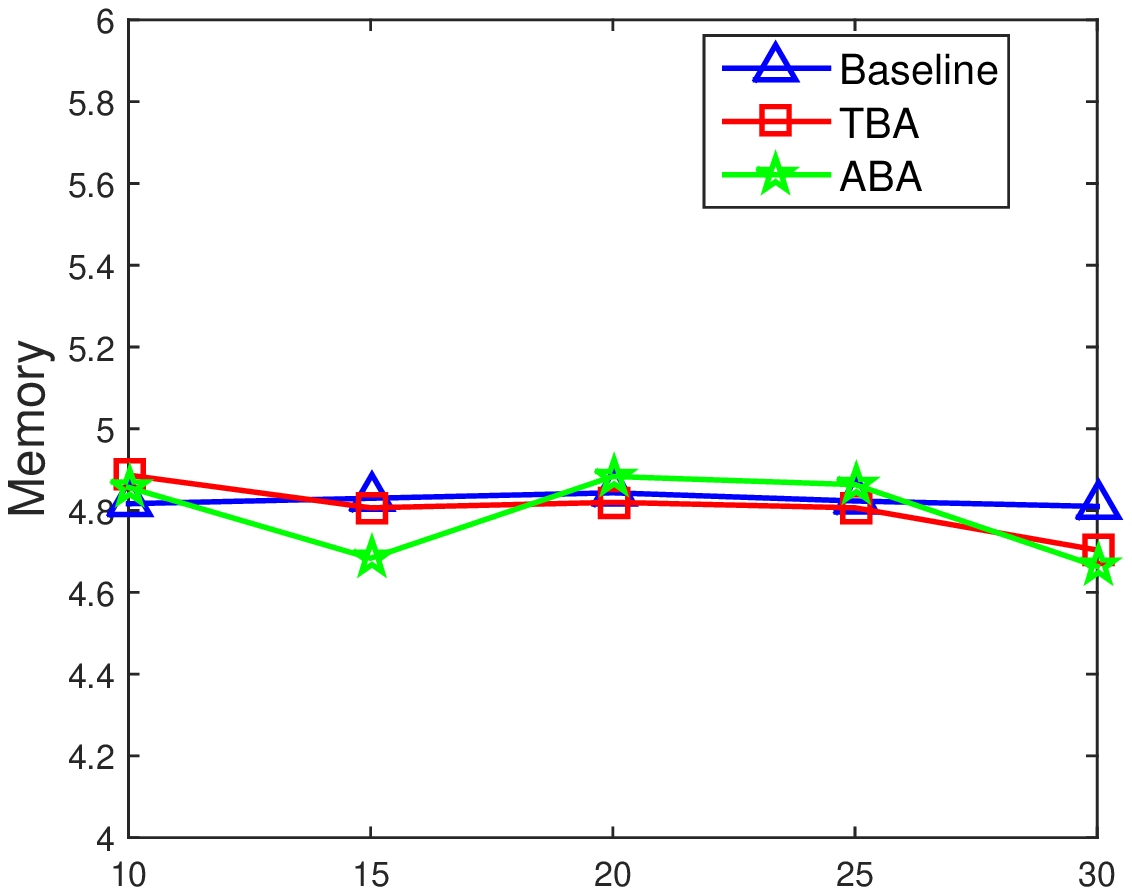}
		\label{fig:result2f}
		%\end{minipage}
	}
	\caption{Results on varying $P_w$.}
\end{figure*}

\begin{figure*}[htb]
	\centering
	\subfloat[\scriptsize{Cardinality of varying $|S|$}]{
		%\begin{minipage}{5cm}
		\centering
		\includegraphics[scale=0.3]{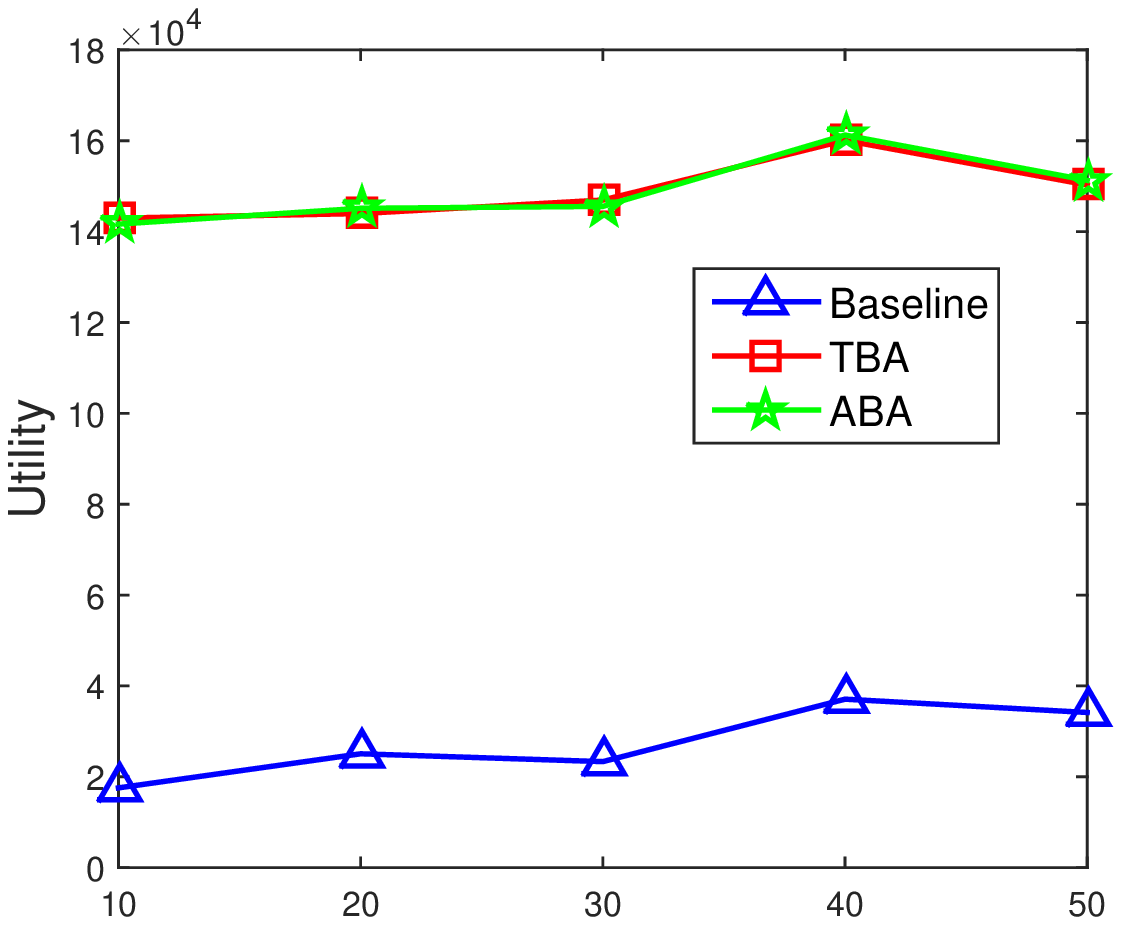}
		\label{fig:result2g}
		%\end{minipage}
	}
	\subfloat[\scriptsize{Runnint Time of varying $|S|$}]{
		%\begin{minipage}{5cm}
		\centering
		\includegraphics[scale=0.3]{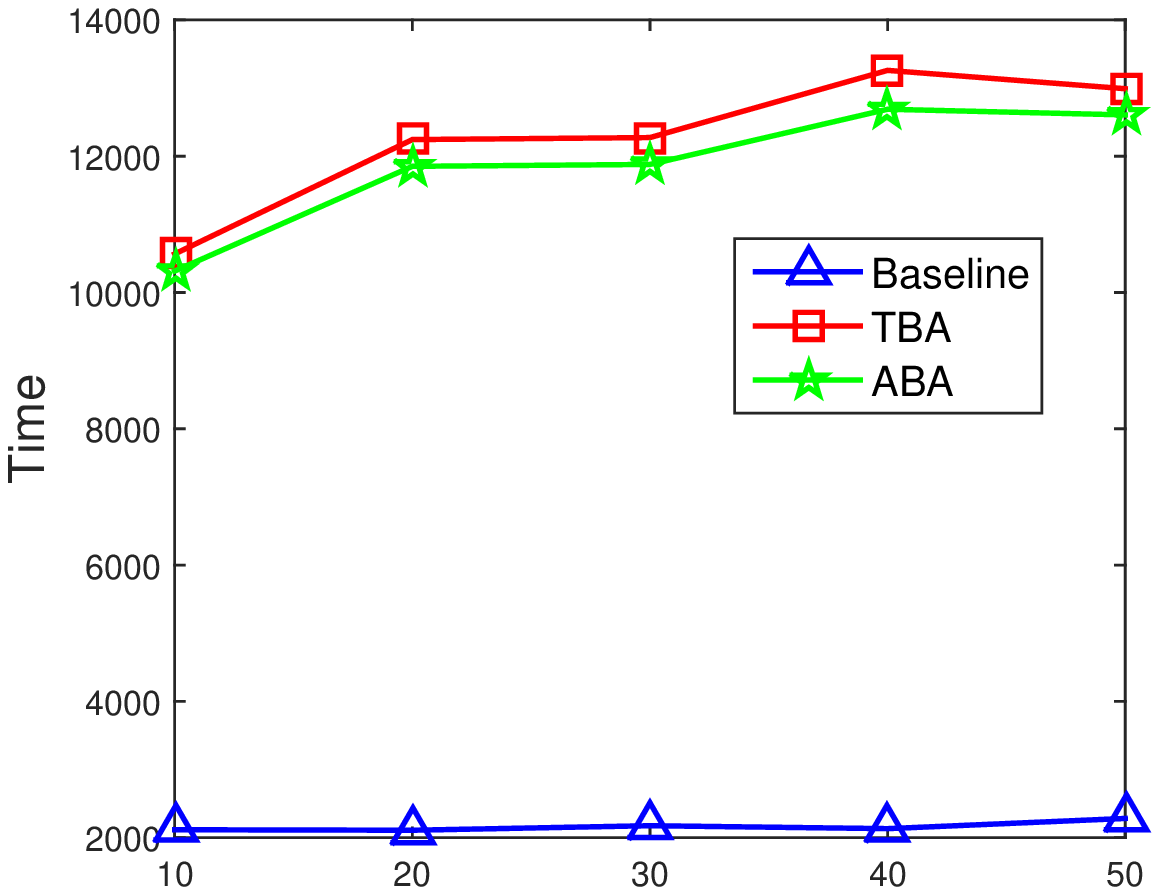}
		\label{fig:result2h}
		%\end{minipage}
	}
	\subfloat[\scriptsize{Memory of varying $|S|$}]{
		%\begin{minipage}{5cm}
		\centering
		\includegraphics[scale=0.3]{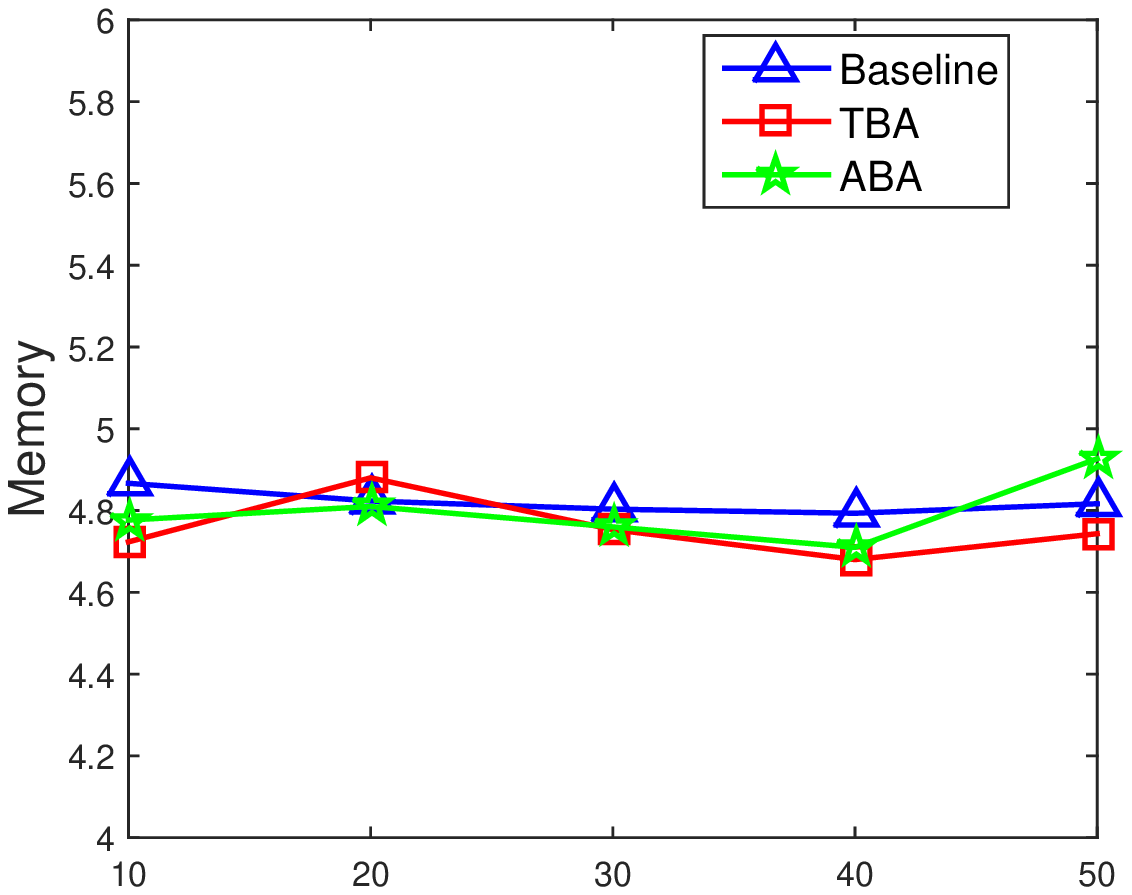}
		\label{fig:result2i}
		%\end{minipage}
	}
	
	\caption{Results on varying $|S|$.}
	\label{fig:result2}
\end{figure*}

\textbf{Effects of the number of tasks $|T|$.} The results of varying $|T|$ are presented in Fig.\ref{fig:result1a} to \ref{fig:result1c}. First, we can observe that the utility increases as $|T|$ increases, which is reasonable as more tasks available. Also, we can observe that TBA algorithm and ABA algorithm are much better than baseline algorithm and TBA algorithm has advantages over ABA algorithm. As for running time, TBA and ABA are slower than the baseline due to sorting tasks and finding more economic schedule, and the running time is acceptable for better performance on utility. Moreover, TBA is faster than ABA for it is easier to find suitable workers for each tasks. The three algorithm do not vary much in memory consumption.

\textbf{Effects of the number of workers $|W|$.} The results of varying $|W|$ are presented in Fig.\ref{fig:result1d} to \ref{fig:result1f}. We can observe that the utility, running time and memory consumption generally increase as $|W|$ increase, which is reasonable as more workers need to be assigned. Again, we can see that TBA are better than ABA in terms of Utility and running time.

\begin{figure*}[htb]
	% Requires \usepackage{graphicx}
	\centering
	\subfloat[\scriptsize{Cardinality of varying $|T|$}]{
		%\begin{minipage}{5cm}
		\centering
		\includegraphics[scale=0.3]{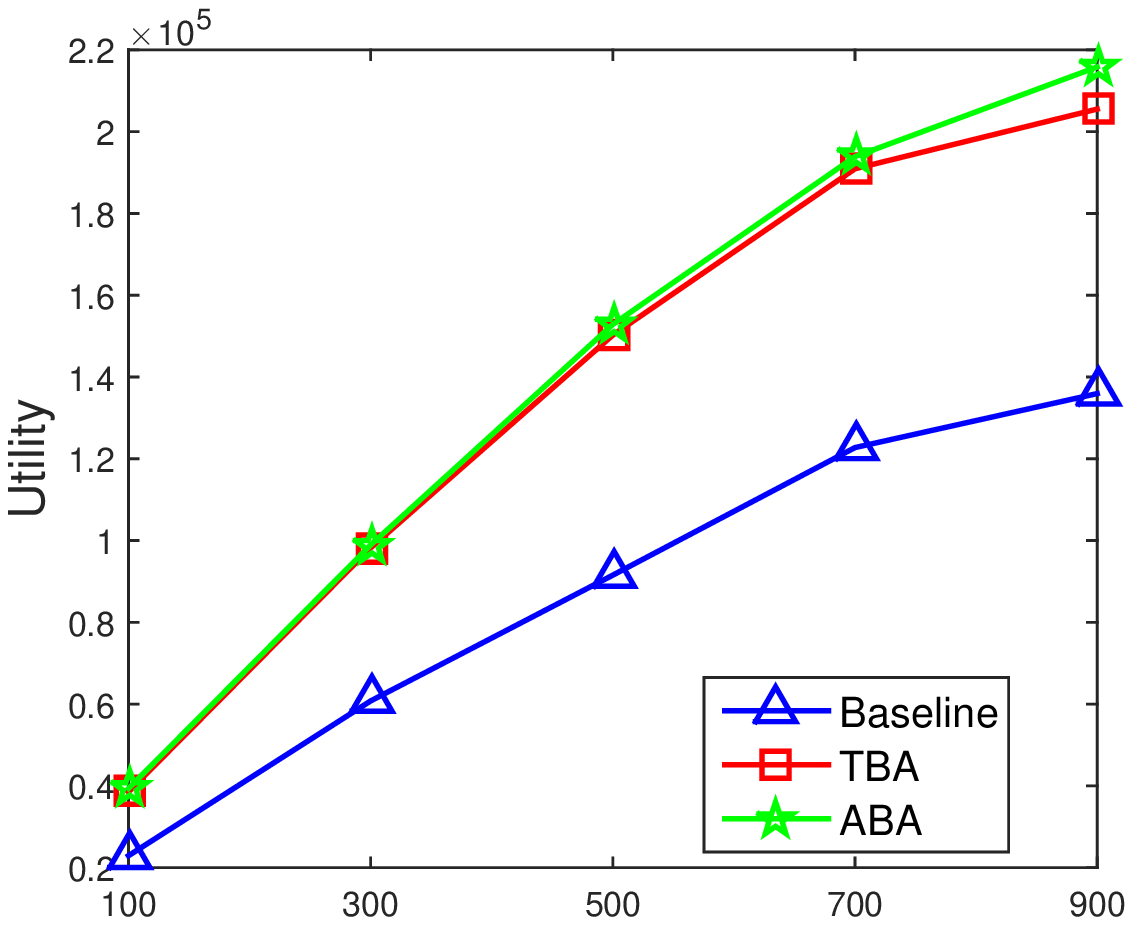}
		\label{fig:result3a}
		%\end{minipage}
	}
	\subfloat[\scriptsize{Running Time of varying $|T|$}]{
		%\begin{minipage}{5cm}
		\centering
		\includegraphics[scale=0.3]{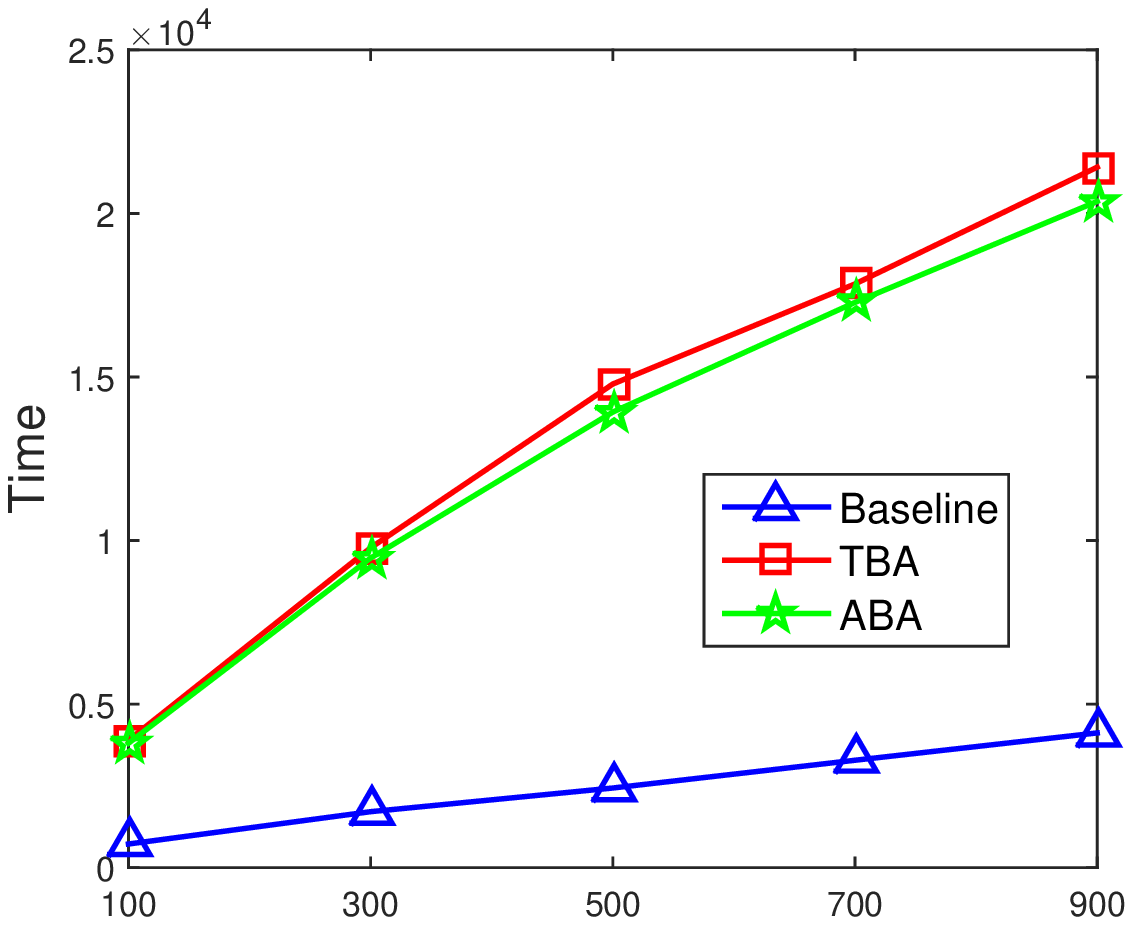}
		\label{fig:result3b}
		%\end{minipage}
	}
	\subfloat[\scriptsize{Memory of varying $|T|$}]{
		%\begin{minipage}{5cm}
		\centering
		\includegraphics[scale=0.3]{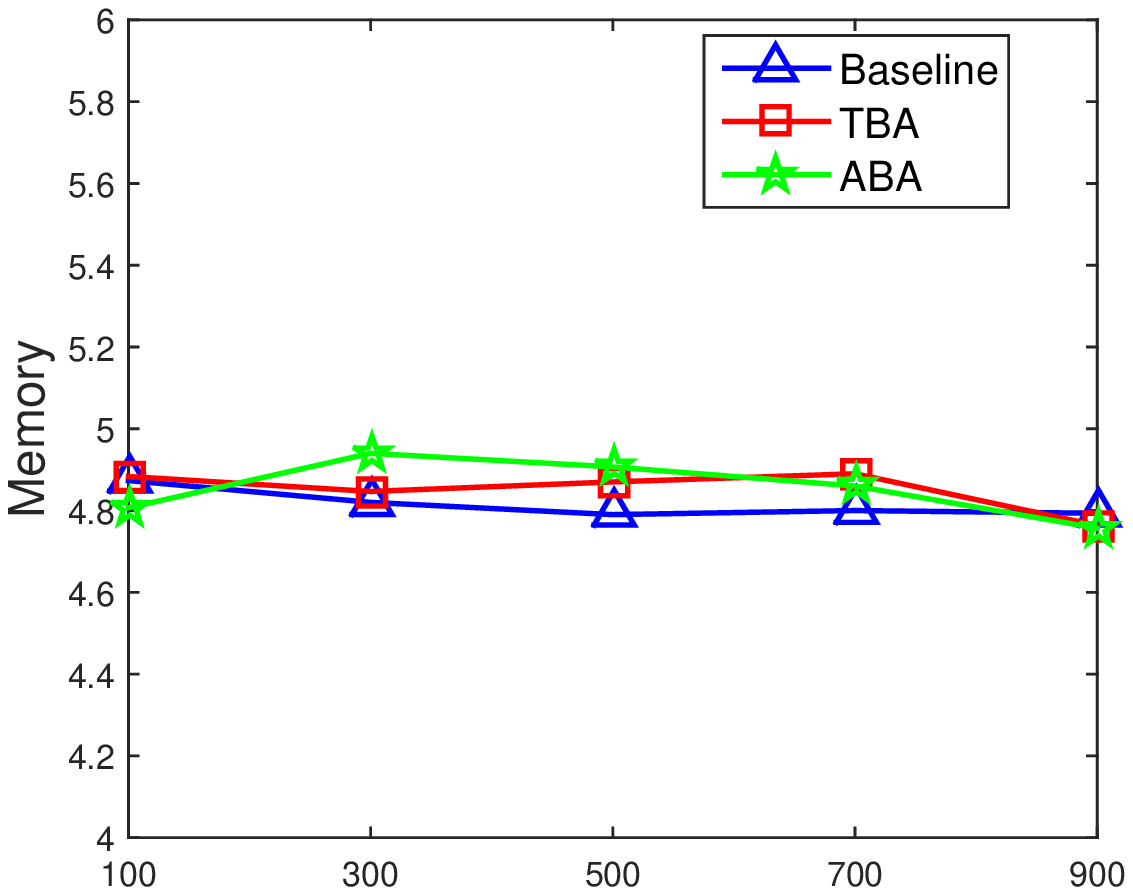}
		\label{fig:result3c}
		%\end{minipage}
	}
	\caption{Results of Real Dataset on varying $|T|$.}
	%\label{fig:result3}
\end{figure*}

\begin{figure*}[htb]
	% Requires \usepackage{graphicx}
	\centering
	\subfloat[\scriptsize{Cardinality of varying $\gamma$}]{
		%\begin{minipage}{5cm}
		\centering
		\includegraphics[scale=0.3]{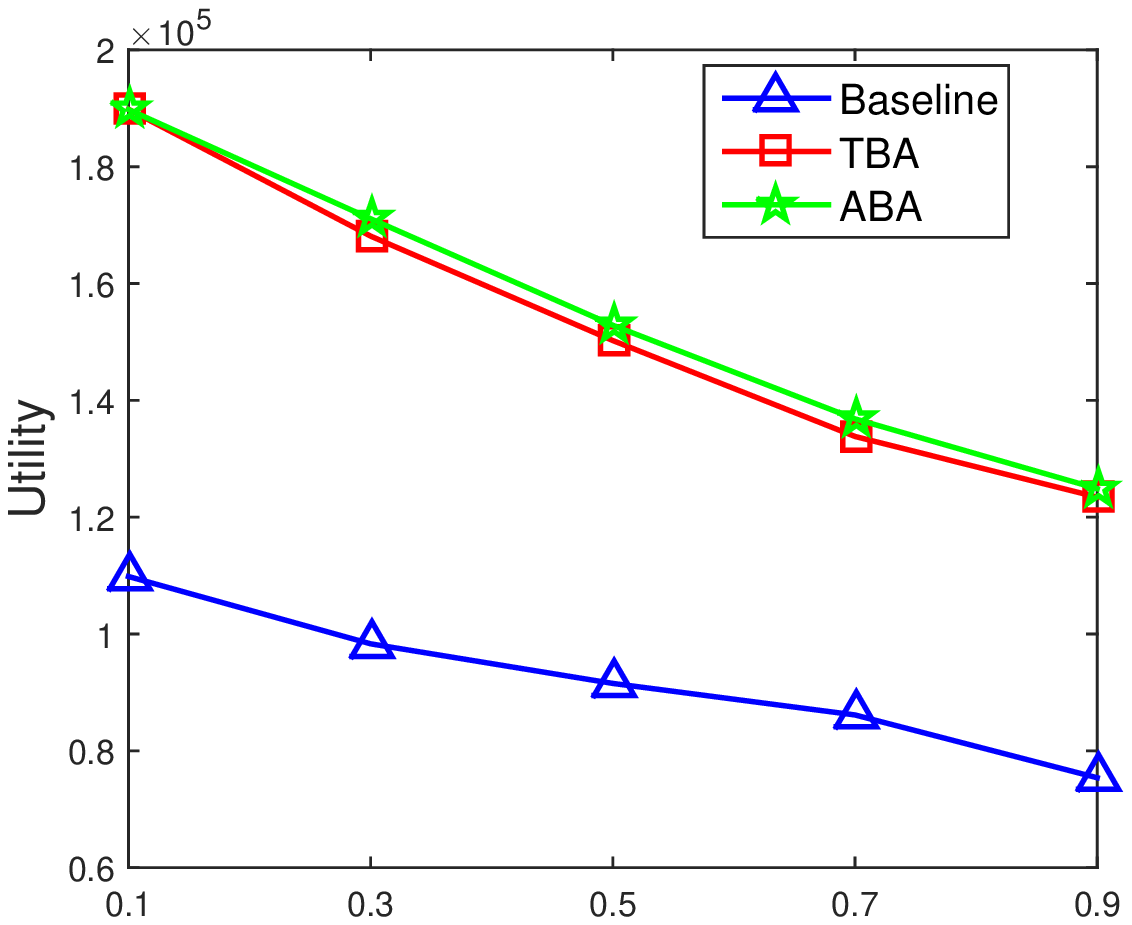}
		\label{fig:result3d}
		%\end{minipage}
	}
	\subfloat[\scriptsize{Running Time of varying $\gamma$}]{
		%\begin{minipage}{5cm}
		\centering
		\includegraphics[scale=0.3]{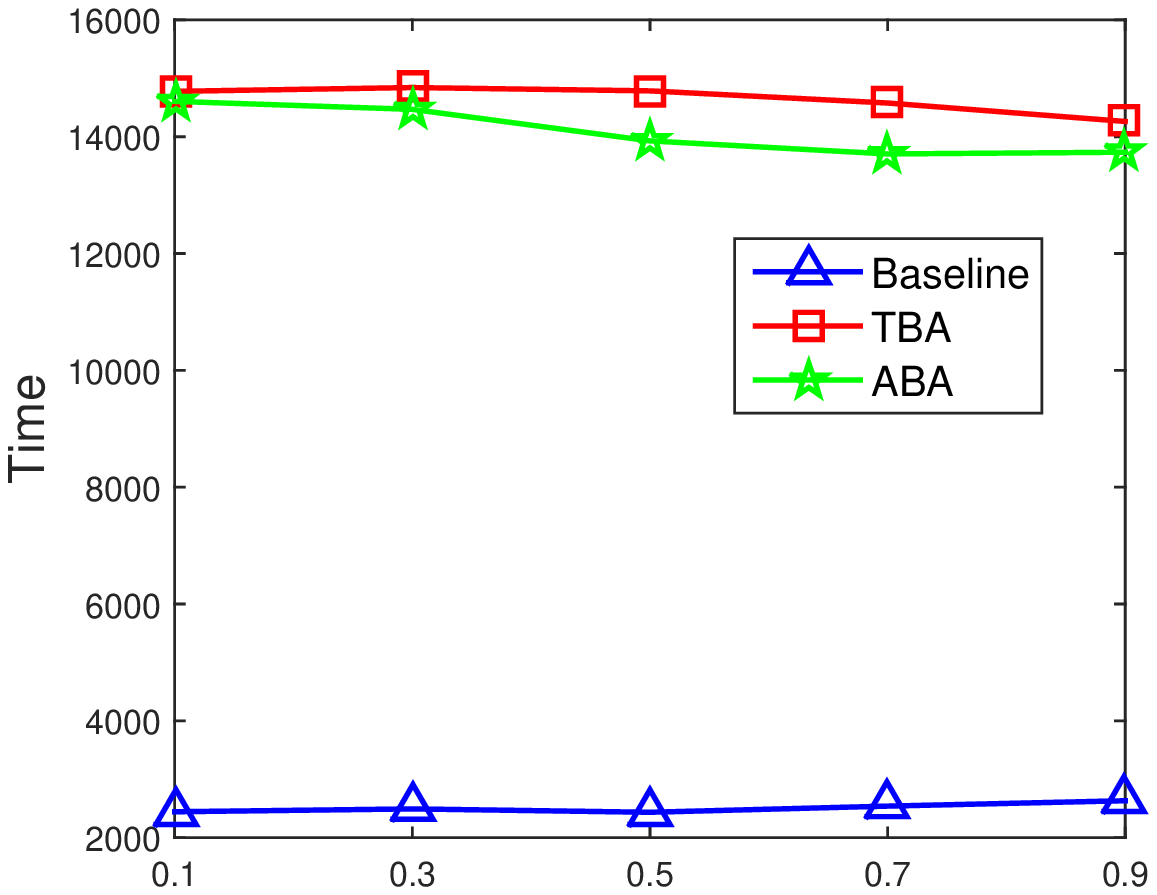}
		\label{fig:result3e}
		%\end{minipage}
	}
	\subfloat[\scriptsize{Memory of varying $\gamma$}]{
		%\begin{minipage}{5cm}
		\centering
		\includegraphics[scale=0.3]{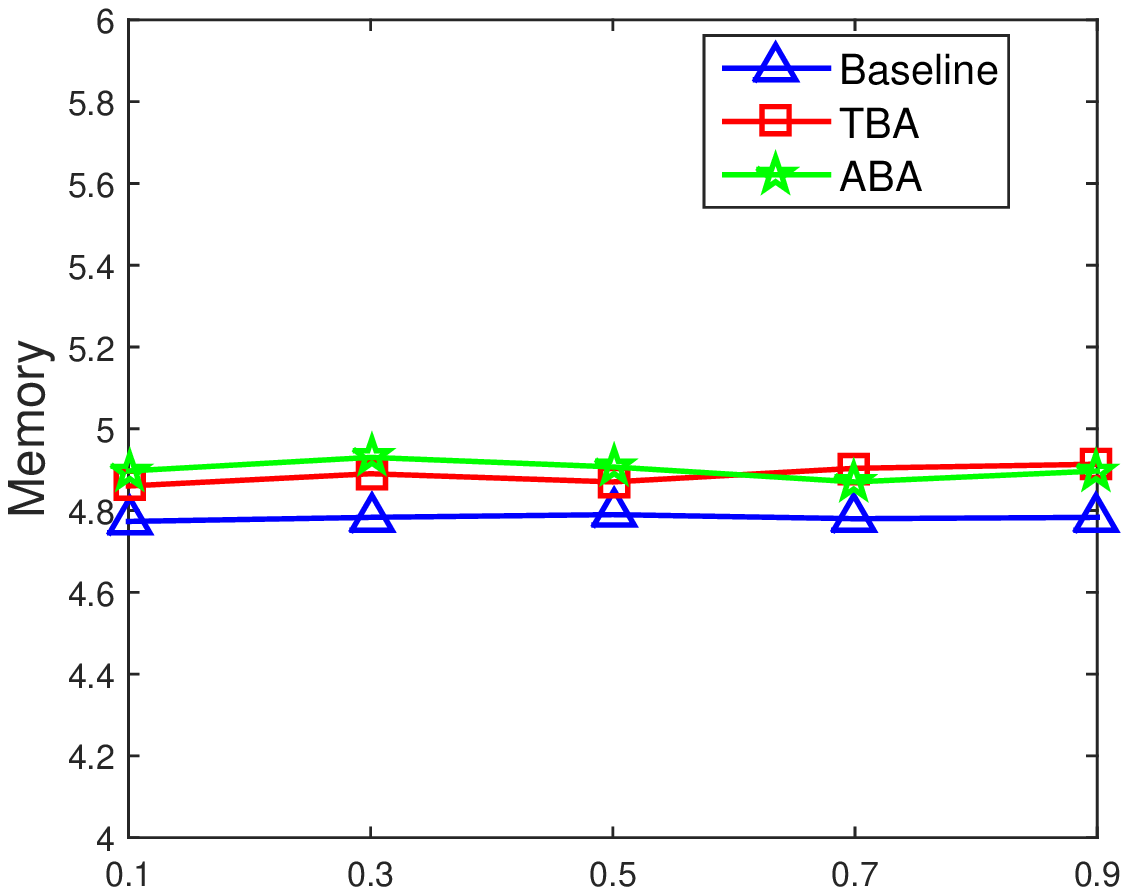}
		\label{fig:result3f}
		%\end{minipage}
	}
	
	\caption{Results of Real Dataset on varying $\gamma$.}
	\label{fig:result3}
\end{figure*}

\textbf{Effects of the global unit transportation fee $\gamma$.} The results of varying $\gamma$ are presented in Fig.\ref{fig:result1g} to \ref{fig:result1i}. We can see that the utility and running time decrease as the $\gamma$ increases for higher transportation fee and less workers that could be assigned to far tasks.

\textbf{Effects of the average budget of tasks $B_t$.} The results are presented in Fig.\ref{fig:result2a} to \ref{fig:result2c}. We can first see from the figure that the utility increases as the average budget increases. And there is no large differences of the running time and memory consumption between various $B_t$.

\textbf{Effects of the variance of the price of different skills $P_w$.} The results are presented in Fig.\ref{fig:result2d} to \ref{fig:result2f}. We can see from the figures that TBA algorithm and ABA algorithm have much better performance than baseline algorithm as the price increases. And the running time and memory consumption do not vary too much in different price.

\textbf{Effects of the total number of skills $|S|$.} The results are presented in Fig.\ref{fig:result2g} to \ref{fig:result2i}. First, we can observe that the utility and memory consumption do not change greatly as the number increases. Then, we can see that the running time increases as the number increases, and this is reasonable because it is much harder to find suitable workers to finish the task for more kinds of skills.

\textbf{Real Dataset.} The results on real dataset are shown in Fig.\ref{fig:result3a} to \ref{fig:result3f}, where we vary $|T|$ and price. We can observe similar patterns as those in Fig.\ref{fig:result1a} to \ref{fig:result1c} and Fig.\ref{fig:result1g} to \ref{fig:result1i}.

\textbf{Conclusion.} For Utility, TBA is better than ABA and baseline algorithm, and both TBA and ABA algorithm have a much better performance than baseline algorithm. As for running time, baseline algorithm is fastest, but the speed of TBA and ABA algorithm is acceptable for most circumstances. Moreover, TBA algorithm is faster than ABA algorithm.

\section{Related Work}
In this section, we review related works from two categories, namely task assignment and team formation problem.

\subsection{Task Assignment in Spatial Crowdsourcing}
The research on task assignment in spatial crowdsourcing mainly includes two parts: micro-task assignment and specialty-aware task assignment. 

Micro task refers to the spatial tasks that can be completed by any single worker. 
\cite{GIS12} is the first work on task assignment in spatial crowdsourcing, whose optimization objective is to maximize the total number of the assignment tasks. 
\cite{ICDE16} is the first work focusing on the online scenario of task assignment, and studies the two-sided online task assignment problem, whose goal is to maximize the total utility score of the assignment. 
\cite{ICDE17} also focuses on the online scenario and considers the influence of work space on task assignment, whose goal is to maximize the total utility score. 
\cite{VLDB16} studies the problem of online minimum weighted bipartite matching, which can be used in online task assignment. 
\cite{VLDB17} considers the problem of flexible online matching where workers can be scheduled if no task is assigned. 
\cite{DASFAA18Tao} recommends routes dynamically for workers to deal with online tasks, and the goal is to maximize the total utility.
\cite{ICDE18Zeng} assigns tasks to workers while trading off quality and latency of task completion.
\cite{SIGMOD18Tong} proposes a match-based approach to solve the dynamic pricing problem in spatial crowdsourcing.
\cite{CIKM17Zhao} takes the destinations of workers into consideration to perform task assignment.
\cite{ICDE18To} considers performing online task assignment while preserving the privacy of tasks and workers under the circumstance that the server is untrusted.
\cite{to2016real,tran2018real} proposes a real-time framework for task assignment.
The difference between our work and the aforementioned works is that they focus on micro tasks which can be completed by a single worker, and we study on the assignment for specialty-aware tasks which have requirements on skills of workers and usually have to be completed by multiple workers collaboratively.

\cite{WAIM16,DSE17Gao} recommend top-k teams with the minimum cost to a specialty-aware task. \cite{TKDE16} studies assigning workers for specialty-aware tasks to maximize the total utility score. The difference between our work and \cite{TKDE16} is that in our work workers specify fees for each of their skills, and in \cite{TKDE16} workers only have a united fee, which is not practical.

\subsection{Team Formation Problem}
A closely related topic is the team formation problem \cite{KDD09TeamFormation}, whose goal is to find a team of experts with the minimum cost, according to the skills and social relationships of the users. \cite{WWW12OnlineTeamFormation} studies the online version of the team formation problem, where the issue of workload balance is also considered. \cite{KDD12CapacitatedTeamFormation} studies another variant of the team formation problem where the capacity constraint of experts is considered. The difference between our problem and the team formation problem and its variants is that we do not consider the social relationships between users and focus on task assignment.
\section{Conclusion}
In this paper we study the problem of \underline{S}pecialty-\underline{A}ware \underline{T}ask \underline{A}ssignment (SATA) in spatial crowdsourcing, where the tasks have requirements on skills, and the workers specify fees for each of their skills. The goal is to maximize the total utility of the task assignment between tasks and workers. We prove the SATA problem is NP-hard. To solve the problem, we propose two efficient and effective heuristic algorithms. We conduct extensive experiments on both synthetic and real-world datasets to evaluate our algorithms. The experiment results show that our solutions are efficient and effective.

%\balance
\vspace{-3ex}
\bibliographystyle{splncs03}
\bibliography{bibsample}

\end{document}